\documentclass[11pt,a4paper]{article}
\usepackage{jheppub}

\arxivnumber{2505.17182}

\usepackage{hyperref}
\hypersetup{
 colorlinks,
 citecolor=blue,
 linkcolor=magenta,
 urlcolor=blue}
\usepackage[T1]{fontenc}
\usepackage[utf8]{inputenc}
\usepackage{graphicx}
\graphicspath{ {./images/} }
\usepackage{amsmath}
\usepackage{amssymb}
\usepackage{mathrsfs}
\usepackage{dsfont}
\usepackage{mathtools}
\usepackage{braket}
\usepackage{float}
\usepackage{parskip}
\usepackage{enumerate}
\usepackage{array}
\usepackage{setspace}
\usepackage{xfrac}

\usepackage{amsthm}
\newtheoremstyle{thm}{10pt}{10pt}{\itshape}{}{\bfseries}{.}{.5em}{}
 
\theoremstyle{thm}
\newtheorem{theorem}{Theorem}[section]
\theoremstyle{thm}

\theoremstyle{thm}
\newtheorem{prop}[theorem]{Proposition}
\theoremstyle{thm}

\theoremstyle{thm}

\numberwithin{equation}{section}

\newtheoremstyle{def}{10pt}{10pt}{}{}{\bfseries}{.}{ }{\thmname{#1}\thmnumber{ #2}\thmnote{ (#3)}}
\theoremstyle{def}
\newtheorem{definition}[theorem]{Definition}

\newtheoremstyle{rem}{10pt}{10pt}{}{}{\itshape \bfseries}{.}{.5em}{}

\theoremstyle{rem}

\theoremstyle{rem}

\newcommand{\Z}{\mathds{Z}}
\newcommand{\R}{\mathds{R}}
\newcommand{\C}{\mathds{C}}
\newcommand{\Q}{\mathds{Q}}
\newcommand{\N}{\mathds{N}}
\newcommand{\slz}{\mathrm{SL}(2,\mathds{Z})}
\newcommand{\gam}{\begin{pmatrix}
		a & b \\ c & d
	\end{pmatrix}}
\newcommand{\hh}{\mathds{H}}
\DeclareMathOperator{\Img}{Im}
\DeclareMathOperator{\Rel}{Re}

\DeclareMathOperator{\Tr}{Tr}

\newcommand{\del}{\partial}

\def\CC {{\cal C}}

\def\CH {{\cal H}}

\def\CM {{\cal M}}
\def\CN {{\cal N}}

\def\CH {{\cal H}}

\def\CT {{\cal T}}

\def\IM{\mathds{M}}
\def\IN{\mathds{N}}

\def\IR{{\mathds{R}}}

\def\IZ{{\mathds{Z}}}

\begin{document}
\title{\boldmath Mock modularity of twisted index in CHL models}


\author[a]{Nabamita Banerjee,}
\author[a,b]{Vedant Bhutra}
\author[c]{and Ranveer Kumar Singh}

\affiliation[a]{Department of Physics, Indian Institute of Science Education and Research Bhopal,\\ Bhopal Bypass Road, Bhopal MP 462066, India}
\affiliation[b]{Department of Theoretical Physics, Tata Institute of Fundamental Research, \\
1 Homi Bhabha Road, Mumbai MH 400005, India}
\affiliation[c]{NHETC and Department of Physics and Astronomy, Rutgers University,\\ 126
Frelinghuysen Rd., Piscataway NJ 08855, USA\\}

\emailAdd{nabamita@iiserb.ac.in}
\emailAdd{vedant.bhutra@tifr.res.in}
\emailAdd{ranveersfl@gmail.com}

\abstract{We study the twisted partition function of quarter BPS states in CHL models and show that for a large class of single-centered black holes, the degeneracy of microstates is given by the Fourier coefficients of mock Jacobi forms. Our analysis is a continuation of the programme initiated by Dabholkar, Murthy and Zagier (DMZ) for $1/4$-BPS dyons in $\mathcal{N} = 4$ string theory and further extended by Bhand, Sen and Singh (BSS) to quarter BPS states in CHL models. We also present the multiplicative lift construction of the partition function and comment on the additive lift of the same.}

\maketitle
\section{Introduction}
Black holes are one of the most interesting outcomes of general relativity as they present perfect laboratories to test quantum aspects of gravity. There are various models in string theory that are useful for studying the microscopic physics of black holes. In particular, the four dimensional CHL models are hugely successful in this aspect. CHL models are type IIB string theory compactified on $\mathcal{M} \times S^1 \times \widetilde{S}^1$ where the compactification manifold $\mathcal{M}$ is either K3 or $T^4$. The resultant theory is a four dimensional $\mathcal{N}=4$ supersymmetric string theory \cite{chaudhuri1995maximally,chaudhuri1995moduli}. In its weak coupling limit, ie. in its simplest form, a sub-sector of the theory can be thought of as product of the following mutually non-interacting microscopic systems: 
\begin{enumerate}
\item a Kaluza-Klein monopole charge associated with circle $\widetilde{S}^1$,
\item a D5-brane charge wrapped on $\mathcal{M} \times S^1$,
\item $Q_1$ units of D1-brane charge wrapped on $S^1$, and 
\item $n$ units of left-moving momentum along $S^1$
and $J$ units of momentum along $\widetilde{S}^1$. 
\end{enumerate}
In the strong coupling limit, the above configuration gives rise to a dyonic black hole carrying electric and magnetic charge vectors $\vec{Q}$ and $\vec{P}$ such that \cite{sen2008black}, \begin{equation}
		Q^2 = 2n,  \qquad P^2 = 2 Q_1, \qquad Q \cdot P = J, \qquad n, Q_1, J \in \mathds {Z}.\nonumber
        \end{equation}
	As we describe next, one gets further exotic configurations by orbifolding the above system. Let us consider a special subspace of the moduli space of $\mathcal{M}$, where the theory has a $\mathds{Z}_M \times \mathds{Z}_N$ symmetry. We consider the following two situations:

	{ \bf Case I. Twisted-twined:} We modify the configuration by taking the left-moving momentum to be $n/M$ along $S^1$ and only considering a  $\mathds{Z}_M $ orbifold of the theory leaving out the $\mathds{Z}_N$ symmetry. We then count the states invariant under the $\mathds{Z}_N$ symmetry.

	{\bf Case II. Doubly twined:} We modify the configuration by taking the left-moving momentum to be $n/M$ along $S^1$ and $J/N$ along $\widetilde{S}^1$ and considering the $\mathds{Z}_M \times \mathds{Z}_N$ orbifold of this theory.
    
In both of the above cases, the $\mathds{Z}_M \times \mathds{Z}_N$ symmetry commutes with the 16 supersymmetry generators of the theory. Furthermore, at the zero temperature, they correspond to quarter BPS states and hence counting of the dyonic black hole degeneracies $d(\vec{Q},\vec{P})$, corresponding to particular charges $(\vec{Q},\vec{P})$ simply boils down to the computation of the supersymmetric index of the corresponding BPS states. The supersymmetry of the system ensures that the index computed in the weak coupling limit is invariant of the string coupling, and therefore the results hold even at strong coupling. Following usual statistical physics, the black hole entropy $S_\text{BH}$ can be then computed from the logarithm of the degeneracies as
\begin{equation}
    S_\text{BH}(\vec Q,\vec P)= \ln d(\vec{Q},\vec{P}). \nonumber
\end{equation}
The quantity computed agrees with the Bekenstein-Hawking entropy of black holes in the large charge limit \cite{strominger1996microscopic,dijkgraaf1997counting,jatkar2006dyon,sen2008black}. The equivalence holds beyond the large charge limit to a finer precession in inverse and exponential suppression in charges, and gives a good understanding of black hole microstates. The degeneracy functions of single-centered dyonic black holes appearing in these models have been successfully computed\footnote{Barring the doubly twisted case, where exact expressions remain to be obtained.}. Moreover, many interesting aspects of the degeneracy function that was predicted from AdS$_2$/CFT$_1$ correspondence have been explicitly tested with these models \cite{Mandal:2010cj}. We list below the prominent results that are already known for the single-centered quarter BPS dyonic black holes of this system.

\begin{itemize}
    \item The degeneracy $d(\vec{Q},\vec{P})$ of $\mathds{Z}_N$-invariant states in $\mathds{Z}_M $ CHL models is given in terms of a specific Fourier transform of the inverse of a function $\Phi^{(M,N)}_k$, where the Fourier integral is defined over a specified contour in the moduli space \cite{sen2010discrete}. We shall present detailed expression of the function in section \ref{sec:twist_twined_par_fu} of the paper.
    \item The function $\Phi^{(M,N)}_k $ is a Siegel modular form of weight $k$ for a subgroup $G$ of\footnote{We follow the convention that $\mathrm{Sp}(2n,\mathds{F})$ is the group of $2n \times 2n$ matrices over $\mathds{F} \in \{\C, \R, \Z\}$ that preserve the symplectic form. Some older references, such as \cite{sen2010discrete}, denote the same group by $\mathrm{Sp}(n,\mathds{F})$.} $\mathrm{Sp}(4,\mathds {Z})$. Its form has been obtained by multiplicative lifts of appropriate seed Jacobi forms \cite {dijkgraaf1997counting,david2006product,jatkar2006dyon} for all values of $(M,N)$ and by additive lift for $(M,N)=(2,2)$, $(2,4)$, and $(4,4)$ \cite{govindarajan2011bkm}. For $N=1$, the Fourier-Jacobi coefficients of ${\Phi_k^{(M,1)}}^{-1}$ were shown to be meromorphic Jacobi forms and a canonical decomposition of these meromorphic Jacobi forms as a sum of a mock Jacobi form and an Appell-Lerch sum was proved in {\cite {dabholkar2012quantum, bhand2023mock}}.
    \item The quantum degeneracies of single-centered black holes arising in $\Z_M$ CHL models are Fourier coefficients of this mock Jacobi form. The degeneracies of multi-centered black holes which decay upon wall-crossing are given by the  Appell-Lerch sum {\cite {dabholkar2012quantum, bhand2023mock}}. Interestingly, this decomposition does not work for all the charge vectors corresponding to a single-centered black hole. There are some special ones for which the degeneracy cannot directly be related to the Fourier coefficients of mock Jacobi forms without using $S$-duality {\cite{bhand2023mock}}.
    \item The mock modularity of the generating functions of certain single-centered black hole degeneracies is useful in obtaining exact formulas for degeneracies in terms of Rademacher sums, as shown for the $(M,N)=(1,1)$ case in \cite{cardoso2005asymptotic}. One can also use the methods of number theory to get a better handle on the growth of the degeneracy using the mock modular properties. For $(M,N)=(1,1)$ case, the mock modular behaviour was useful in proving Sen's positivity conjecture \cite{sen2011black}, i.e.,
\begin{equation}
    d(n,\ell,m)>0,\qquad \text{for } m,n>0,\,4mn-\ell^2>0.
\end{equation}
This was proved for $m=1,2$ in \cite{Bringmann:2012zr} and for $m\geq 3$ in \cite{rossello2024immortal}.  
\end{itemize}

The above success certainly demands further investigations for the twisted-twined and the doubly twined case. In the present paper, we restrict ourselves to the former, ie. we consider $\mathds{Z}_M$ orbifolds of this model leaving behind a $\mathds{Z}_N$-symmetry. In the case of $\mathcal{M}=\text{K}3$, allowed $\Z_M \times \Z_N$ symmetries can be worked out using Nikulin's classification of symplectic automorphisms of K3 surfaces. We refer the readers to \cite{chaudhuri1996type,aspinwall1996some} for an extensive list of such symmetries. 

The formulation of the problem is enumerated as follows :
\begin{enumerate}
    \item We begin with the partition function of the twisted-twined sector given in \cite{sen2010discrete}.
    \item Next we find out the transformation properties of the partition function. This is a non-trivial step and we have used the threshold integral representation of the partition function to find the transformation properties. 
    \item Then we apply and extend the results of \cite{bhand2023mock} to prove the mock modularity of finite parts of the Fourier-Jacobi coefficients of the partition function. 
    \item Finally, we relate the degeneracy of the single-centered black holes to mock Jacobi forms.  
\end{enumerate}

Similar to the case studied in \cite{bhand2023mock}, we find that there exists a charge sector for which the degeneracy of the single-centered black holes, appearing in the twisted CHL model, can not be described in terms of a mock Jacobi form. 

The paper is arranged as follows: in section \ref{sec:2}, we present the mathematical prerequisites, as well as summarize (without proofs) the most general versions of the important theorems of \cite{dabholkar2012quantum, bhand2023mock}. The readers who are familiar with the same may skip this section. Section \ref{sec:3} contains the main results of the paper, where we found the mock Jacobi forms that describe the degeneracies of the $\Z_N$-invariant single-centered black holes, that arises in the $\Z_M$ CHL models. In section \ref{sec:4} we present for $\CM=\text{K}3$ the dyon partition function as a multiplicative lift of an appropriate seed function and comment on the existence of an additive lift of the same. Finally we end the paper with a discussion on open problems in section \ref{sec:5}.

\section{Mathematical preliminaries}\label{sec:2}
In this section, we define a few mathematical objects relevant for this paper. For more details and some proofs, we refer the reader to \cite{dabholkar2012quantum, bhand2023mock}. Let $\hh$ be the upper half of the complex plane $\C$ (excluding the real line).

\paragraph{Conventions.} For the rest of the paper, we choose the following conventions,
\begin{align}
    \tau, \sigma &\in \hh, \qquad z \in \C, \\
\bold{e}^x \coloneqq e^{2 \pi i x}, &\qquad q \coloneqq \bold{e}^\tau, \qquad y \coloneqq \bold{e}^z, \qquad p \coloneqq \bold{e}^\sigma.
\end{align}

\subsection{Jacobi and Siegel modular forms}
\begin{definition}[Jacobi form on $\Gamma \ltimes (M\Z \times \Z)$ of weight $k$ and index $m/M$]
\label{jacobi}
	A holomorphic function $\varphi_{k,\frac{m}{M}} : \hh \times \C \to \C$ that satisfies the following:
	\begin{enumerate}
		\item \textit{modularity condition} (for all $\gamma = \gam \in \Gamma$): \begin{equation}
 	 \varphi_{k,\frac{m}{M}}\bigg(\frac{a \tau + b}{c \tau + d}, \frac{z}{c \tau + d} \bigg) = (c \tau + d)^k \, \bold{e}^{\tfrac{m c z^2}{M(c \tau + d)}} \, \varphi_{k,\frac{m}{M}}(\tau, z), \end{equation}
 	\item \textit{elliptic condition} (for $\lambda \in M\Z$, $\mu \in \Z$): \begin{equation}
	\label{ellipticcond}
 		\varphi_{k,\frac{m}{M}}(\tau, z + \lambda \tau + \mu) = \bold{e}^{-\frac{m}{M}(\lambda^2 \tau + 2 \lambda z)} \, \varphi_{k,\frac{m}{M}}(\tau, z),
 	\end{equation}
	\end{enumerate}
     along with some cusp condition, where\footnote{One may also define Jacobi forms of half-integral weight. In that case, $\Gamma \subset \Gamma_0(4)$ and some additional constants appear in the modularity condition.} $k \in \Z$, $m, M \in \Z_{>0}$ and $\Gamma < \slz$. 
\end{definition}

Henceforth, we assume $\Gamma \in \{\Gamma_0(N), \Gamma_1(N)\}$, where
\begin{align}
	\label{heckegrps}
		\Gamma_1(N) &\coloneqq \bigg\{\gam \in \slz : \, \begin{matrix} c \equiv 0 \pmod{N} \\ a \equiv d \equiv 1 \pmod{N} \end{matrix} \bigg\}, \\
		\Gamma^1(N) &\coloneqq \bigg\{\gam \in \slz : \, \begin{matrix} b \equiv 0 \pmod{N} \\ a \equiv  d \equiv 1 \pmod{N} \end{matrix} \bigg\}.
\end{align} are congruence subgroups of $\slz$, sometimes called \textit{Hecke subgroups}.\footnote{$\Gamma_0$ and $\Gamma^0$ are larger Hecke subgroups defined similarly, albeit without the restriction on $a$ and $d$.}
Also,
\begin{equation}
    \Gamma(M,N):=\Gamma_1(N)\cap\Gamma^1(M).
\end{equation}

A Jacobi form as defined above for $\Gamma \in \{\Gamma_0(N), \Gamma_1(N)\}$ is periodic in both $\tau, z$ with period $1$. Therefore it admits a Fourier expansion of the form,
\begin{equation}
\label{Jacobifourierexp}
		\varphi_{k,\frac{m}{M}}(\tau,z) = \sum_{n,r \, \in \, \Z} c(n,r) \, q^n y^r.
\end{equation}

In physics, we generally encounter \textit{weak Jacobi forms} (defined by $c(n,r) = 0$ for $n<0$ in the above expansion), as they have exponential growth of their Fourier coefficients, which relates to the degeneracy of microstates.

We also define a useful function, called the \textit{index $m$ $\vartheta$-function}
\begin{equation}
		\vartheta_{m,\ell}(\tau,z) \coloneqq \sum_{\substack{r \, \in \, \Z \\ r \, \equiv \, \ell (\mathrm{mod} \, 2m)}} q^{r^2/4m} y^r = \sum_{n\, \in \, \Z} q^{(\ell + 2mn)^2 / 4m} y^{\ell + 2mn},
\end{equation} which is a (component of a vector-valued) Jacobi form of weight $\frac{1}{2}$ and index $m \in \Z$. 

Then, one can show that Jacobi forms have a \textit{theta expansion} in terms of these $\vartheta$-functions given by \cite{eichler1985theory, bhand2023mock}
\begin{equation}
\label{thetaexp}
	\varphi(\tau, z) = \sum_{\ell \in \Z/2m\Z} h_\ell(\tau) \, \widetilde{\vartheta}_{\frac{m}{M},\ell}(\tau,z),
\end{equation} where,
\begin{equation}
    \widetilde{\vartheta}_{\frac{m}{M},\ell}(\tau,z) \coloneqq \vartheta_{m,\ell}(M\tau,z).
\end{equation}

The coefficients $h_\ell (\tau)$ also transform as a vector-valued weakly holomorphic modular form $h \coloneqq (h_1, \ldots, h_{2m})$ of weight $k - \frac{1}{2}$ with respect to $\Gamma (4mN)$.
By Cauchy's theorem, we have the following integral formula:
\begin{equation}
\label{thetacoeff}
	h_\ell (\tau) = q^{-M\ell^2/4m} \int_{z_0}^{z_0 + 1} \varphi(\tau,z) \, y^{-\ell}\, dz,
\end{equation} where $z_0 \in \C$ is arbitrary, and the left hand side is independent of $z_0$ and the path of integration.

We also define \textit{Jacobi-Erderlyi $\vartheta$-functions},
\begin{align}
	\vartheta_1(\tau, z) &= -i \sum_{n \in \Z + \frac{1}{2}}(-1)^{n-\frac{1}{2}} y^n q^\frac{n^2}{2}, \\
	\vartheta_2(\tau, z) &= \sum_{n \in \Z + \frac{1}{2}} y^n q^\frac{n^2}{2}, \\
	\vartheta_3(\tau, z) &= \sum_{n \in \Z} y^n q^\frac{n^2}{2}.
\end{align}

Then, one can construct
\begin{equation}
\label{def:phi01}
    \varphi_{0,1}(\tau, z) = 4 \bigg( \frac{\vartheta_2^2(\tau, z)}{\vartheta_2^2(\tau, 0)} + \frac{\vartheta_3^2(\tau, z)}{\vartheta_3^2(\tau, 0)} + \frac{\vartheta_4^2(\tau, z)}{\vartheta_4^2(\tau, 0)} \bigg),
\end{equation} which is the unique Jacobi form of weight 0 and index 1 for $\slz$. This Jacobi form is (up to a scaling factor) the elliptic genus of K3.

The above discussion on Jacobi forms can be generalized to meromorphic functions. It is shown in \cite{dabholkar2012quantum, bhand2023mock} that certain meromorphic Jacobi forms over $\Gamma \ltimes \Z^2$ can be decomposed into finite and polar pieces, where the finite part is a mixed mock Jacobi form and the polar part is constructed from Appell-Lerch sums. We shall describe this in detail in the next subsection.

For now, suppose $\varphi(\tau,z)$ is a \textit{meromorphic} Jacobi form of weight $k$ and index $m/M$ with respect to $\Gamma \ltimes \Z^2$ such that it has only simple or double poles, and only at $z$ belonging to some discrete subset of $\Q \tau + \Q$. We adopt the following notation for the position of the poles:
	$S(\varphi) \coloneqq \{s = (\alpha, \beta) \in \Q^2 : z = z_s = \alpha \tau + \beta \text{ is a pole of $\varphi$}\}$.

As $\varphi$ is meromorphic, it may not admit a theta decomposition where the coefficient $h_\ell(\tau)$ is independent of $z_0$ and the path of integration, as in (\ref{thetaexp}), as it could jump when the integration contour is deformed across a pole. Periodicity in $\ell$ is also typically destroyed. We get around this by defining,
\begin{equation}
    \label{merthetacoeff}
	h^\star _\ell(\tau) \coloneqq q^{M\ell^2 / 4m} \int_{\R / \Z} \varphi(\tau, z - \ell M \tau/2m) \, y^{-\ell} \, dz,
\end{equation} which is effectively fixing $z_0 = -\frac{\ell M}{2m} \tau$ in (\ref{thetacoeff}). Any poles in the path are avoided by taking averages of the integrals just above and below the poles. We shall quickly see how this is beneficial.

For more mathematical details on Jacobi forms, we refer readers to \cite{eichler1985theory}. We now turn our attention to Siegel modular forms. We define the genus-2 generalisation of $\hh$, the \textit{Siegel upper half plane} $\hh_2$ by
	\begin{equation}
    \label{siegelupperhalfplane}
		\hh_2 \coloneqq \bigg\{ \Omega = \begin{pmatrix}
			\tau & z \\ z & \sigma
		\end{pmatrix} : \tau, \sigma \in \hh, z \in \C, \,\det(\Img(\Omega)) > 0 \bigg\}.
	\end{equation} Note that we will follow this convention for $\Omega$ for the rest of the paper.

We now have all the ingredients to define,
\begin{definition}[Siegel modular form on $G < \mathrm{Sp}(4,\Z)$ of weight $k$] A holomorphic function $\Phi:\hh_2 \to \C$ such that,
	\begin{equation}
		\Phi(\Omega) \to \Phi\big((A\Omega + B)(C \Omega + D)^{-1}\big) = \det(C\Omega + D)^k \, \Phi(\Omega), \qquad \forall \begin{pmatrix}
	A & B \\ C & D
\end{pmatrix} \in G.
	\end{equation}
\end{definition}

Depending on $G$, the SMF satisfies periodicity in the three variables due to modular invariance. In particular, for $G = \mathrm{Sp}(4,\Z)$, $\tau, z, \sigma$ all have unit period and we have the Fourier expansion,
\begin{equation}
	\Phi(\Omega) = \sum_{\substack{n, m, r \in \Z\\m,n\geq 0 \\ r^2 \leq 4mn}} c(n,m,r) \, q^n p^m y^r.
\end{equation} 
The restriction $m,n\geq 0,r^2\leq 4mn$ in the sum over the indices is because of holomorphicity of $\Phi$. It may be restricted differently depending on the analytic properties of $\Phi$.

In particular, we may write,
\begin{equation}
	\Phi(\Omega) = \sum_{m} \phi^F_m(\tau, z) \, p^m,
\end{equation} where $\phi^F_m$ is easily shown to be a Jacobi form of weight $k$ and index $m$. Thus, this is called the \textit{Fourier-Jacobi expansion} of $\Phi$. The reason for the superscript $F$ will be made clear in the next subsection.

The prototypical example of a Siegel modular form is the \textit{Igusa cusp form},
\begin{equation}
	\label{igusa}
		\Phi_{10} (\Omega) \coloneqq q p y \prod_{\substack{m,n,\ell \in \Z :\\ m> 0, n\geq 0 \text{ or}, \\  m\geq 0, n> 0 \text{ or}, \\ m = n = 0, \, \ell < 0.}} (1 - q^n p^m y^\ell)^{2 c_0(4nm - \ell^2)},
	\end{equation} which is a weight 10 Siegel modular form, where $c_0$ are the Fourier coefficients of the $\varphi_{0,1}$ (defined in \ref{def:phi01})). For a more mathematically rigorous review, we refer readers to \cite{pitale2019siegel, bruinier20081}.

\subsection{Mock modular forms}
Finally, we turn to mock modular forms. These are holomorphic functions that \textit{almost} satisfy the modularity property, but not quite. Instead, they have a non-holomorphic \textit{completion}, that is modular. Mock-modular forms fave found numerous applications in number theory, combinatorics, moonshine and theoretical physics. A good review on the subject is \cite{folsom2017perspectives}, which we shall follow.

\begin{definition}[Pure mock modular form]
\label{puremock}
	A \textit{(weakly holomorphic) pure mock modular form of weight $k \in \frac{1}{2}\Z$} is the first member of the pair $(h,g)$ such that,
	\begin{enumerate}
		\item $h$ is a holomorphic function on $\hh$ with at most exponential growth at all cusps,
		\item $g$, called the \textit{shadow of $h$}, is a holomorphic modular form of weight $2-k$.
		\item $\widehat{h} \coloneqq h + g^*$, called the \textit{completion of $h$}, transforms like a holomorphic modular form of weight $k$.
	\end{enumerate}
\end{definition}

The function $g^*$ is called the \textit{non-holomorphic Eichler integral}, and is a solution to the differential equation,
\begin{equation}
	(4 \pi \Img(\tau))^k \, \frac{\del g^*}{\del \overline{\tau}} = - 2 \pi i \, \overline{g}.
\end{equation}

Being a modular form, $g$ has a Fourier expansion, given by $g(\tau) = \sum_{n \geq 0} a_n q^n$. Thus, we can fix $g^*$ by setting
\begin{equation}
g^*(\tau) = \begin{cases} 
      -\overline{a}_0 \dfrac{(4 \pi \Img(\tau))^{-k + 1}}{-k+1} + \sum_{n > 0} n^{k-1} \, \overline{a}_n \, \Gamma(1-k, 4 \pi n \Img(\tau)) \, q^{-n}, & k \neq 1 \\
      -\overline{a}_0 \ln(4 \pi \Img(\tau)) + \sum_{n > 0} \, \overline{a}_n \, \Gamma(0, 4 \pi n \Img(\tau)) \, q^{-n} & k = 1
   \end{cases}
\end{equation} where $\Gamma(a,b)$ is the incomplete gamma function:
\begin{equation}
    \Gamma(a, b):=\int_b^{\infty} e^{-t} t^a \frac{d t}{t}.
\end{equation}

Holomorphicity in $h$ implies that $\widehat{h}$ is related to the shadow $g$ by,
\begin{equation}
	(4 \pi \Img(\tau))^k \, \frac{\del \widehat{h}}{\del \overline{\tau}} = - 2 \pi i \, \overline{g}.
\end{equation}

\begin{definition}[Pure mock Jacobi form]
	A \textit{pure mock Jacobi form of weight $k$ and index $m$} is a holomorphic function $\varphi$ on $\hh \times \C$ that satisfies the elliptic transformation property (\ref{ellipticcond}) and the cuspidal condition, such that the theta coefficients of $\varphi$ are mock modular forms of weight $k - \frac{1}{2}$. The \textit{completion} $\widehat{\varphi}$ of $\varphi$ is defined as    
\begin{equation}
	\widehat{\varphi} (\tau, z) = \sum_{\ell \in \Z/2m\Z} \widehat{h}_\ell (\tau) \, \vartheta _{m, \ell}(\tau, z)=\varphi (\tau, z) + \sum_{\ell \in \Z/2m\Z} g^*_\ell (\tau) \vartheta_{m,\ell} (\tau, z).
\end{equation} where $g_\ell$ is the shadow of $h_\ell$. This is a Jacobi form of weight $k$ and index $m$.
\end{definition}

This is a natural generalization of pure mock modular forms. However, in most practical applications, we have more complicated functions as defined below. A prominent example will be the Appell-Lerch sums that we define in the next section.

\begin{definition}[Mixed mock modular forms]
	A \textit{mixed holomorphic mock modular form of weight $k|\vec{\ell}$}, with $k, \ell_j \in \frac{1}{2}\Z$, is a holomorphic function $h(\tau)$ with polynomial growth at cusps and completion of the form $\widehat{h} = h + \sum_j f_j g_j^*$, where $f_j$ is a holomorphic modular form of weight $\ell_j$ and $g_j$ is a holomorphic modular form of weight ${2-k+\ell_j}$. The completion transforms like a modular form of weight $k$.
\end{definition}

Finally, we have \textit{mixed mock Jacobi forms}, which are defined analogous to the above, ie. they have theta expansions in which the coefficients are mixed mock modular forms,
with completions that transform like ordinary Jacobi forms.

\subsection{Finite and polar parts}
We return to the discussion of meromorphic Jacobi forms. Right away, we can define a \textit{finite part} of a meromorphic Jacobi form $\varphi$ as
\begin{equation}
\label{finite}
	\varphi^F(\tau, z) \coloneqq \sum_{\ell \in \Z/2m\Z} h^\star_\ell (\tau) \, \widetilde{\vartheta}_{\frac{m}{M},\ell}(\tau,z),
\end{equation} with $h^\star_\ell$ as defined above in (\ref{merthetacoeff}). What remains is to construct a function that has poles exactly at $S(\varphi)$, weighted with the correct residues, that also satisfies the desired invariance properties. We will call this the \textit{polar part} of $\varphi$.

As we have restricted the class of meromorphic Jacobi forms to those with simple or double poles at specific points as described earlier, we can define the residues $D_s$ and $E_s$ by the Laurent expansion
\begin{equation}
	\bold{e}^{\frac{m}{M} \alpha z_s} \, \varphi(\tau, z_s + \epsilon) = \frac{E_s (\tau)}{(2 \pi i \epsilon)^2} + \frac{D_s (\tau) - 2\frac{m}{M}\alpha E_s(\tau)}{2 \pi i \epsilon} + \cdots \qquad \text{as } \epsilon \to 0.
\end{equation}

Then for $m \in \Z$, and fixed $q$ in the unit disc on $\C$, the \textit{averaging operator} is given by
\begin{equation}
	\text{Av}^{(m)}_q\big[f(y)\big] \coloneqq \sum_{\lambda \, \in \, \Z} q^{m \lambda^2} y^{2m \lambda} f(q^\lambda y),
\end{equation} for any function $f(y)$. It outputs a function that satisfies the elliptic property of an index $m$ Jacobi form. In particular, one can easily check that \begin{equation}
	\vartheta_{m,\ell}(\tau,z) = q^\frac{\ell^2}{4m} \, \text{Av}_q^{(m)}[y^\ell].
\end{equation}

For $c \in \R$ and $z \in \C \backslash \Z$, we  define the \textit{$n^\text{th}$-order rational function} by
\begin{equation}
	\mathcal{R}_c^{(n)} (y) \coloneqq \frac{y^c}{(2\pi i)^n}\sum_{r \in \Z} \frac{\bold{e}^{-cr}}{(z-r)^n}.
\end{equation}
It has a simple pole of residue 1 at $y = 1$. We finally have all the ingredients to define the \textit{$n^\text{th}$-order Appell-Lerch sum of index $m/M$} for $s \in \Q^2$ as
	\begin{equation}
		\widetilde{\mathcal{A}}^{\,n,s}_\frac{m}{M}(\tau, z) \coloneqq \mathcal{A}^{n,\widetilde{s}}_m(M\tau, z), \qquad \mathcal{A}^{n,s}_m(\tau, z) \coloneqq \bold{e}^{- m \alpha z_s} \text{Av}_q^{(m)}\bigg[\mathcal{R}^{(n)}_{-2 m \alpha}\bigg(\frac{y}{y_s}\bigg) \bigg],
	\end{equation} where $\widetilde{s} = (\alpha/M,\beta)$. It turns out that the Appell-Lerch sums are meromorphic mock Jacobi forms of weight 1 and index $m/M$ with respect to some congruence subgroup of $\slz$ depending on $s$.

Then, the \textit{polar part} of $\varphi$ is defined by
	\begin{equation}
    \varphi^P(\tau, z) \coloneqq \sum_{s \, \in \, S(\varphi)/(M\Z \times \Z)} \big(D_s(\tau) \, \widetilde{\mathcal{A}}_\frac{m}{M}^{\,1,s} (\tau, z) + E_s(\tau) \, \widetilde{\mathcal{A}}_\frac{m}{M}^{\,2,s} (\tau, z) \big).
	\end{equation} Now, we can state without proof,
\begin{theorem}[\cite{dabholkar2012quantum,bhand2023mock}]
\label{dmzmaintheorem}
	Let $\varphi(\tau, z)$ be a meromorphic Jacobi form with simple poles at $z_s$ for all $s \in S(\varphi)$, some discrete subset of $\Q^2$,  then $\varphi$ has the decomposition
	\begin{equation}
	\varphi(\tau, z) = \varphi^F(\tau, z) + \varphi^P(\tau, z),
\end{equation} where $\varphi^F$ and $\varphi^P$ are as defined above.
\end{theorem}

Note that, so far, there is no reason to make the particular choice of $\varphi^F$ and $\varphi^P$ as we did. We could have defined $h_\ell^\star (\tau)$ differently, leading to a different splitting. However, our choice is special because of the particular transformation properties of $\varphi^F$ we obtain. First, we define the Eichler integrals of unary theta series
\begin{align}
\label{thetaeichler}
	\Theta^{s*}_{\frac{m}{M},\ell} (\tau) &= \frac{\bold{e}^{\frac{m}{M} \alpha \beta}}{2 \sqrt{\pi}} \sum_{\lambda \, \in \, \Z + \frac{\alpha}{M} + \ell/2m} \text{sgn}(\lambda) \, \bold{e}^{-2 m \beta \lambda} \, \Gamma\bigg(\frac{1}{2}, 4 \pi m M\lambda^2 \Img(\tau) \bigg) q^{-m M \lambda^2}, \\
    \label{unarytheta2}
	\vartheta^{s*}_{\frac{m}{M},\ell} (\tau) &= \frac{\bold{e}^{\frac{m}{M} \alpha \beta}}{2 \pi \sqrt{mM  \Img(\tau)}} - \frac{\bold{e}^{\frac{m}{M} \alpha \beta}}{\sqrt{\pi}} \sum_{\lambda \, \in \, \Z + \frac{\alpha}{M} + \ell/2m} |\lambda| \, \bold{e}^{-2 m \beta \lambda} \, \Gamma\bigg(-\frac{1}{2}, 4 \pi m M \lambda^2  \Img(\tau) \bigg) q^{-mM \lambda^2}.
\end{align}

Then, we have the following theorem:
\begin{theorem}[\cite{dabholkar2012quantum,bhand2023mock}]
\label{dmzmaintheorem2}
	Let $\varphi$ be a meromorphic Jacobi form of weight $k$ and index $m/M$ with respect to $\Gamma \ltimes \Z^2$ having simple poles on some subset $S(\varphi) \subset \Q^2$ as described earlier. Let $h_\ell^\star$, $\Theta_{\frac{m}{M}, \ell}^{s*}$, $\vartheta_{\frac{m}{M}, \ell}^{s*}$ and $\varphi^F$ be as defined in (\ref{merthetacoeff}), (\ref{thetaeichler}), (\ref{unarytheta2}) and (\ref{finite}) respectively. Then,
	\begin{enumerate}
		\item $h_\ell^\star$ is a mixed mock modular form with completion:
	\begin{equation}
		\widehat{h}_\ell^\star (\tau) \coloneqq h_\ell^\star (\tau) - \sum_{s \, \in \, S(\varphi)/(M\Z \times \Z)} \big( D_s(\tau) \, \Theta_{\frac{m}{M}, \ell}^{s*} (\tau) + E_s(\tau) \vartheta^{s*}_{\frac{m}{M},\ell} (\tau) \big).
	\end{equation} 
	\item $\varphi^F$ is a mixed mock Jacobi form with the completion,
	\begin{equation}
    \label{eq:finitepartcompletion}
		\widehat{\varphi}^F(\tau, z) \coloneqq \sum_{\ell \, \in \, \Z/2m\Z} \widehat{h}^\star_\ell (\tau) \, \widetilde{\vartheta}_{\frac{m}{M},\ell}(\tau,z),
	\end{equation} which itself is a Jacobi form of weight $k$ and index $\frac{m}{M}$ with respect to $\Gamma \ltimes (M\Z \times \Z)$.
	\end{enumerate}
\end{theorem}

Theorems \ref{dmzmaintheorem} and \ref{dmzmaintheorem2}, and their application to finding degeneracies of $1/4$-BPS dyons in CHL models, are the main results of \cite{dabholkar2012quantum, bhand2023mock}. In the following section, we shall use these theorems to find the degeneracies of microstates using the twisted index.

\section{Mock Jacobi forms from twisted partition function in CHL models}\label{sec:3}
As stated in the introduction, we are interested in the degeneracy of certain single-centered dyons in CHL models, i.e. quarter BPS states in $\Z_M$ CHL models invariant under the residual $\Z_N$-symmetry. Let us describe the microscopic system in detail.

\paragraph{The system.} We focus on states carrying one unit of Kaluza-Klein monopole charge associated with circle $\widetilde{S}^1$, one unit of D5-brane charge wrapped on $\mathcal{M} \times S^1$, $Q_1$ units of D1-brane charge wrapped on $S^1$, left-moving momentum $n/M$ along $S^1$ and $J$ units of momentum along $\widetilde{S}^1$. Note the convention that the coordinate radius of the original circle $S^1$ before orbifolding is taken to be $2 \pi M$. Thus the minimum amount of momentum along $S^1$ is $1/M$.

We define electric and magnetic charge vectors $\vec{Q}$ and $\vec{P}$ such that
\begin{equation}
\label{eq:chargevec}
    Q^2 = \frac{2n}{M}, \qquad P^2 = 2 Q_1, \qquad Q \cdot P = J, \qquad n, Q_1, J \in \Z.
\end{equation} 

\subsection{Twisted-twined partition function} \label{sec:twist_twined_par_fu}
The partition function of quarter-BPS dyons invariant under $\Z_N$ for $\Z_M$ orbifolds of type IIB superstring theory compactified on $\mathcal{M} \times S^1 \times \widetilde{S}^1$ for $\mathcal{M} = \text{K}3$ or $T^4$ first appeared in \cite{sen2010discrete}. It is the inverse of a Siegel modular form $\Phi^{(M,N)}_k : \hh_2 \to \C$ of weight $k$ of a subgroup $G$ of $\mathrm{Sp}(4,\Z)$, given by the infinite product expansion,
\begin{align}\label{eq:product_PhiMN}
	\Phi^{(M,N)}_k(\tau, z, \sigma) &= q^{\alpha} y^{\beta} p^{\gamma} \prod_{b=0}^1 \prod_{r=0}^{N-1} \prod_{r'=0}^{M-1} \prod_{\substack{k' \in \Z_{\geq 0} + \frac{r'}{M} \\ \\ \ell \in \Z_{\geq 0}, \, j \in 2 \Z + b \\ \\j < 0 \text{ for } k'=\ell=0}} \big[1 - \bold{e}^{r/N} q^{\ell} y^j p^{k'}\big]^a, \\
	a &\coloneqq \sum_{s=0}^{N-1} \sum_{s'=0}^{M-1} c_b^{(0,s;r',s')}(4 k' \ell - j^2) \, \bold{e}^{-(r \frac{s}{N}+\ell \frac{s'}{M})},
\end{align} where the coefficients $c_b^{(r,s;r',s')}$ are defined via
\begin{equation}\label{eq:cbrsr's'def}
	\sum_{b=0}^1\sum_{\substack{\ell \in 2 \Z + b \\ n \in \Z/MN}} c_b^{(r,s;r',s')} (4n - \ell^2) \, q^n y^\ell = \frac{1}{MN} \Tr_{\text{RR}; \, g_M^{r'} g_N^r} \big[g_M^{s'} g_N^s \, (-1)^{J_L + J_R} \, q^{L_0-\frac{c}{24}} \bar{q}^{\bar{L}_0-\frac{\bar{c}}{24}} y^{J_L} \big],
\end{equation} 
where $g_M,g_N$ are the generators of $\Z_M$ and $\Z_N$ respectively,  $\Tr_{\text{RR}; \, g_M^{r'} g_N^r}$ means that the trace is taken over all the $g_M^{r^{\prime}} g_N^r$ twisted RR sector states in the $c=\bar{c}=6$, $\CN=(4,4)$ SCFT with target space $\mathcal{M}$, $L_0$ and $\bar{L}_0$ are the left and right-moving Virasoro generators and $J_L / 2$ and $J_R / 2$ are the generators of the $\mathrm{U}(1)_L \times \mathrm{U}(1)_R$ subgroup of the $\mathrm{SU}(2)_L \times \mathrm{SU}(2)_R$ R-symmetry group of this SCFT, 
and
\begin{align}
	\alpha &= \frac{1}{24M} Q_{0,0} - \frac{1}{2M} \sum_{s = 1}^{M-1} Q_{0,s'} \frac{\bold{e}^{-s/M}}{(1-\bold{e}^{-s/M})^2}, \label{eq:alpha_def}\\
	\beta &= 1, \\
	\gamma &= \frac{1}{24M}\chi(\mathcal{M})=\frac{1}{24M} Q_{0,0}, \label{eq:gamma_def}
\end{align} 
where $\chi(\CM)$ is the Euler number of manifold $\CM$, $Q_{r',s'} \coloneqq Q_{0,0;r',s'}$ and
\begin{equation}
  Q_{r,s;r',s'} \coloneqq MN\left(c_0^{(r,s;r',s')}(0) + 2c_1^{(r,s;r',s')}(-1)\right).  
\end{equation}

One can show that $\alpha,M\gamma\in\Z$. For $N = M = 1$, the partition function is exactly the inverse of the Igusa cusp form $\Phi_{10}$, defined in (\ref{igusa}) \cite{dijkgraaf1997counting}. 

Some special generators of the group $G$ are found in appendix \ref{app:threshold_integral}. The coefficients $c_b^{(r,s;r',s')}(u)$ satisfy the following \cite{sen2010discrete}:
\begin{align}
\label{eq:c0(u)u<0} c_0^{(r,s;r',s')}(u) &= 0,\qquad u<0, \\
\label{eq:c1(u)u<-1} c_1^{(r,s;r',s')}(u) &= 0,\qquad u<-1,
\end{align}
and,
\begin{align}
\label{eq:c1(-1)} c_1^{\left(0, s ; 0, s^{\prime}\right)}(-1) &= \begin{cases}
    \frac{2}{M N} &\qquad \text { for } \mathcal{M}=\text{K}3, \\
    \frac{1}{M N}\left(2-\bold{e}^{s / N} \bold{e}^{s^{\prime} / M} - \bold{e}^{-s / N} \bold{e}^{-s^{\prime} / M}\right) &\qquad \text { for } \mathcal{M}=T^4.
\end{cases}
\end{align}

We also have the relation,
\begin{equation}\label{eq:Qrsr's'sum}
    \sum_{b=0}^1\sum_{\ell\in 2\Z+b}c_b^{(r,s;r',s')}(4n-\ell^2)=\frac{1}{MN}Q_{r,s;r',s'}\delta_{n,0},
\end{equation}

\begin{proof}
    Note that taking $z=0$ in \eqref{eq:cbrsr's'def} gives 
\begin{equation}\label{eq:inteq_sumcbrsr's'}
\sum_{b=0}^1 \sum_{\substack{\ell \in 2 \Z + b \\ n \in \Z/MN}} c_b^{(r,s;r',s')} (4n-\ell^2) \, q^n = \frac{1}{MN} \Tr_{\text{RR}; \, g_M^{r'} g_N^r} \big[g_M^{s'} g_N^s \, (-1)^{J_L + J_R} \, q^{L_0-\frac{c}{24}} \bar{q}^{\bar{L}_0-\frac{c}{24}} \big].    
\end{equation}
The Cartan generators $J_L,J_R$ of the $\mathrm{SU}(2)_L \times \mathrm{SU}(2)_R$ R-symmetry group of the non-linear sigma model with target space $\mathcal{M}$ can be identified with the worldsheet fermion numbers $F_L,F_R$ respectively \cite{sen2008black}. Because of unbroken supersymmetry, the fermionic and bosonic states with $(L_0-c/24,\bar{L}_0-\bar{c}/24)\neq (0,0)$ are paired and hence the terms on the RHS of \eqref{eq:inteq_sumcbrsr's'} with $(L_0-c/24,\bar{L}_0-\bar{c}/24)\neq (0,0)$ cancel. 
Consequently, the $\tau$-dependence of the RHS of \eqref{eq:inteq_sumcbrsr's'} drops out and only the $n=0$ term on the LHS of \eqref{eq:inteq_sumcbrsr's'} contributes. The result now follows using \eqref{eq:c0(u)u<0} and \eqref{eq:c1(u)u<-1}.
\end{proof} 

Also note that, for $r, r', s, s' = 0$, the RHS of \eqref{eq:cbrsr's'def} is just the elliptic genus of $\mathcal{M}$ and is given by
\begin{align}
\bold{EG}=
	\begin{cases}
		 0 &\text{ for } \mathcal{M} = T^4, \\
		 \frac{2}{MN} \, \varphi_{0,1}(\tau, z) &\text{ for } \mathcal{M} = \text{K}3,
	\end{cases} 
\end{align} where $\varphi_{0,1}$ is defined in (\ref{def:phi01}).

For non-vanishing $r, r', s, s'$, the contributions can be computed using purely geometric data -- the fixed points of elements of $\Z_M \times \Z_N$ and the action of elements of $\Z_M \times \Z_N$ near these points. See \cite[appendix A]{sen2010discrete} for an algorithm for explicitly computing the RHS of (\ref{eq:cbrsr's'def}). 

The weight of the Siegel modular form $\Phi^{(M,N)}_k$
is given by \cite{sen2010discrete}
\begin{equation}
    k:=\frac{1}{2} \sum_{s=0}^{N-1} \sum_{s^{\prime}=0}^{M-1} c_0^{\left(0, s ; 0, s^{\prime}\right)}(0).
\end{equation} 

Finally, the degeneracy of states carrying charges $(\vec{Q},\vec{P})$ is given by
\begin{equation}
    \label{degeneracyformula}
	d(\vec{Q},\vec{P}) = \frac{1}{M}(-1)^{Q \cdot P + 1} \int_\mathcal{C} d\tau dz d\sigma \, q^{-\frac{1}{2}MQ^2} \, y^{-Q \cdot P} \, p^{-\frac{1}{2} P^2 / M} \frac{1}{\Phi^{(M,N)}_k},
\end{equation} where $\mathcal{C}$ is the three dimensional subspace of $\C^3$ given by \cite{Cheng:2007ch}
\begin{align}
\label{eq:attractorcontour}
	\Img(\tau) = \frac{P^2}{M\varepsilon}, \qquad &\Img(z) = -\frac{P\cdot Q}{\varepsilon}, \qquad \Img(\sigma) = \frac{MQ^2}{\varepsilon},\\
	0 \leq \Rel(\tau) \leq 1, \qquad 0 \leq &\Rel(z) \leq 1, \qquad 0 \leq \Rel(\sigma) \leq M,
\end{align}
where $\varepsilon>0$ is a small real number. This is called the \textit{attractor contour}.\footnote{The choice of the integration contour is not unique, and the ambiguity is related to the phenomenon of wall crossing. However, the attractor mechanism ensures that the degeneracies are fixed in terms of the black hole charges, and the fluctuations due to the different values of the asymptotic moduli are sufficiently small that they can be ignored \cite{Cheng:2007ch}.}

The Siegel modular form has double zeros on the hypersurface given by \cite{sen2010discrete}

\begin{equation}
\label{eq:zeroes}
	n_2 (\sigma \tau - z^2) + jz + n_1 \sigma - m_1 \tau + m_2 = 0
\end{equation}
such that 
\begin{align}
	m_1\in M\Z,\qquad m_2, n_1 \in\Z,\qquad n_2 \in &N\Z, \qquad j \in 2\Z + 1, \\ m_1n_1 + m_2n_2 + \frac{j^2}{4} &= \frac{1}{4}.
\end{align}
 Finally, we will need the $z\to 0$ behaviour of the partition function. In this limit, the dominant $z$-dependent term in the product \eqref{eq:product_PhiMN} comes from the $k'=\ell=0,j=-1$ term. In particular, we have to consider only the $r'=0,b=1$ term. We want to analyze 
 \begin{equation}
     y\prod_{r=0}^{N-1}(1-\bold{e}^{r/N}y^{-1})^{a'},\qquad a'=\sum_{s=0}^{N-1}\sum_{s'=0}^{M-1}c_1^{(0,s;0,s')}(-1) \bold{e}^{-rs/N}.
 \end{equation}
 
 Using the identity 
 \begin{equation}
     \sum_{s=0}^{N-1} \bold{e}^{-rs/N}=\begin{cases}
         N,&r\equiv 0\bmod N,
         \\
         0,&\text{otherwise},
     \end{cases}
 \end{equation}
 and the properties \eqref{eq:c1(-1)}, we see that 
 \begin{equation}
     a'=2\delta_{r,0}.
 \end{equation}
 This gives us 
 \begin{equation}
    y \prod_{r = 0}^{N - 1}(1 - \bold{e}^{r/N} y^{-1})^{a'} \simeq -4 \pi^2 z^2 + O(z^4),\qquad z\to 0. 
 \end{equation}
 We can then write 
 \begin{equation}
 \begin{split}
     \Phi^{(M,N)}_k(\tau,z,\sigma)&\simeq -4\pi^2z^2q^{\alpha}p^{\gamma} \prod_{b=0}^1 \prod_{r=0}^{N-1} \prod_{r'=0}^{M-1} \prod_{\substack{k' \in \Z_{\ge 0} + \frac{r'}{M} \\ \\ \ell \in \Z_{\ge 0}, \, j \in 2 \Z + b \\ \\(k',\ell)\neq(0,0)}} \big[1 - \bold{e}^{r/N} q^{\ell}  p^{k'}\big]^{a}+O(z^4),\qquad z\to 0
     \\
     &=-4\pi^2z^2q^{\alpha}p^{\gamma} \prod_{r=0}^{N-1} \prod_{r'=0}^{M-1} \prod_{\substack{k' \in \Z_{\ge 0} + \frac{r'}{M} \\ \\ \ell \in \Z_{\ge 0}, \, (k',\ell)\neq(0,0)}} \big[1 - \bold{e}^{r/N} q^{\ell}  p^{k'}\big]^{a''}+O(z^4),
 \end{split}    
 \end{equation}
 where 
 \begin{equation}
 a''\coloneqq \sum_{s=0}^{N-1} \sum_{s'=0}^{M-1}\sum_{b=0}^1 \sum_{j\in2\Z+b} c_b^{(0,s;r',s')}(4 k' \ell - j^2) \, \bold{e}^{-(r \frac{s}{N} + \ell \frac{s'}{M})},    
 \end{equation}
 which, using \eqref{eq:Qrsr's'sum}, becomes 
 \begin{equation}
     a''=\frac{\delta_{k'\ell,0}}{MN}\sum_{s=0}^{N-1} \sum_{s'=0}^{M-1}Q_{0,s;r',s'} \bold{e}^{-(r \frac{s}{N} + \ell \frac{s'}{M})}.
 \end{equation}
Thus we get 
\begin{equation}\label{eq:PhiMNatz=0}
\Phi^{(M,N)}_k(\tau,z,\sigma)\simeq -4\pi^2z^2f^{(k)}(\tau)g^{(k)}(\sigma)+O(z^4),\qquad z\to 0,    
\end{equation}
where $f^{(k)}$ corresponds to the $k'=0,\ell\geq 1$ terms and $g^{(k)}$ corresponds to the $k'\geq 1,\ell=0$ terms. Explicitly, we have 
\begin{equation}\label{eq:f_kg_kdef}
\begin{split}
    f^{(k)}(\tau)&:=q^\alpha\prod_{r=0}^{N-1}\prod_{\ell=1}^\infty\left[1-\bold{e}^{r/N}q^\ell\right]^{s_{r\ell}},
    \\
    g^{(k)}(\sigma)&:=p^\gamma\prod_{r=0}^{N-1}\prod_{r'=0}^{M-1}\prod_{k' \in \Z_{>0} + \frac{r'}{M}}\left[1-\bold{e}^{r/N}p^{k'}\right]^{t_{r}},
\end{split}    
\end{equation}
where 
\begin{equation}
\begin{split}
    s_{r\ell}&:=\frac{1}{MN}\sum_{s=0}^{N-1} \sum_{s'=0}^{M-1}Q_{0,s;0,s'}\bold{e}^{-(s \frac{r}{N}+s' \frac{\ell}{M})}=\sum_{s=0}^{N-1} \sum_{s'=0}^{M-1}\left(c_0^{(0,s;0,s')}(0)+2c_1^{(0,s;0,s')}(-1)\right)\bold{e}^{-(s \frac{r}{N}+s' \frac{\ell}{M})},
    \\
    t_{r}&:=\frac{1}{MN}\sum_{s=0}^{N-1} \sum_{s'=0}^{M-1}Q_{0,s;r',s'}\bold{e}^{-s \frac{r}{N}}=\sum_{s=0}^{N-1} \sum_{s'=0}^{M-1}\left(c_0^{(0,s;r',s')}(0)+2c_1^{(0,s;r',s')}(-1)\right)\bold{e}^{-s \frac{r}{N}}.
\end{split}    
\end{equation}

\subsection{Invariance groups of $\Phi^{(M,N)}_k$}
To derive the modular properties of the degeneracy function, we first need to find the invariance properties of the Siegel modular form. For this, we describe some special generators of the subgroup $G < \rm{Sp}(4,\R)$ under which $\Phi^{(M,N)}_k$ transforms as a Siegel modular form. 

For $t,P\in\Z_{> 0}$, let us define the following elements of $\rm{Sp}(4,\R)$:
\begin{equation}\label{eq:g_123def}
\begin{aligned}
g_{1}(a, b, cP, d) & \equiv\left(\begin{array}{cccc}
a & 0 & b & 0 \\
0 & 1 & 0 & 0 \\
cP & 0 & d & 0 \\
0 & 0 & 0 & 1
\end{array}\right),\qquad g^1(a,bP,c,d)\equiv\begin{pmatrix} 
1&0&0&0\\0&a&0&bP\\0&0&1&0\\0&c&0&d
\end{pmatrix},
\\
g_4(a,b,cP,d)&\equiv\begin{pmatrix}
a&b&0&0\\cP&d&0&0\\0&0&d&-cP\\0&0&-b&a
\end{pmatrix},\qquad a d-b cP=1, \qquad a, d=1 \bmod P,\\
V_t &\equiv \frac{1}{\sqrt{t}} \begin{pmatrix}
 		0 & t & 0 & 0 \\ 1 & 0 & 0 & 0 \\ 0 & 0 & 0 & 1 \\ 0 & 0 & t & 0
 	\end{pmatrix},\qquad U_t\equiv \begin{pmatrix}
1 & 0 & 0 & 0 \\
0 & 0 & 0 & \frac{1}{t} \\
0 & 0 & 1 & 1 \\
t & -t & 0 & 0
\end{pmatrix},\qquad t\in\IN, \\
g_{3}(\lambda, \mu,\kappa) & \equiv\left(\begin{array}{cccc}
1 & 0 & 0 & \mu \\
\lambda & 1 & \mu & \kappa \\
0 & 0 & 1 & -\lambda \\
0 & 0 & 0 & 1
\end{array}\right),\qquad g^3(\lambda,\mu,\kappa)\equiv\begin{pmatrix}
1&\lambda&\kappa&\mu\\0&1&\mu&0\\0&0&1&0\\0&0&-\lambda&1
\end{pmatrix}, \qquad
\lambda, \mu, \kappa \in \Z.
\end{aligned}
\end{equation} 

The matrices above define the following embeddings
\begin{equation}
\label{eq:sp4_zembeddings}
\begin{split}    g_1:\Gamma_1(P)\longrightarrow \mathrm{Sp}(4,\Z),\qquad g^1:\Gamma^1(P)\longrightarrow \mathrm{Sp}(4,\Z),
\\
g_3:H_3(\Z)\longrightarrow \mathrm{Sp}(4,\R),\qquad g_4:\Gamma_1(P)\longrightarrow \mathrm{Sp}(4,\Z), 
\end{split}
\end{equation}
where $H_3(\Z)$ is the discrete Heisenberg group.
    
Let $P$, $t \in \Z_{>0}$. We define the \textit{paramodular group} $\Gamma_t(P)$ by
 	\begin{equation}
 		\Gamma_t(P) \coloneqq \left\{ \begin{pmatrix}
 			* & *t & * & * \\ * & * & * & * t^{-1} \\ *P & *Pt & * & * \\ *Pt & *Pt & *t & *
 		\end{pmatrix} \in \mathrm{Sp}(4,\Q), \text{ all } * \in \Z \right\}.
 	\end{equation}
$P$ is called the level of $\Gamma_t(P)$. The \textit{full paramodular group} $\Gamma_t$ is the group for which $P=1$.  
The paramodular group has a normal double extension: 
\begin{equation}
   \Gamma^+_t(P) = \Gamma_t(P)\cup \Gamma_t(P) V_t 
\end{equation}
$\Gamma^+_t(P)$ is generated by the \textit{maximal parabolic subgroup} $\Gamma^\infty_t(P)$ and $V_t$ \cite{Gritsenko:1996tm}, where 
\begin{equation}
    \Gamma_t^{\infty}(P)=\left\{\left(\begin{array}{cccc}
* & 0 & * & * \\
* & 1 & * & * t^{-1} \\
* P & 0 & * & * \\
0 & 0 & 0 & 1
\end{array}\right) \in \Gamma_t(P), \text { all } * \in \mathds{Z}\right\}.
\end{equation}
The \textit{Jacobi group} is then given by 
\begin{equation}
    \Gamma^J(P)=\left(\Gamma_t^{\infty}(P) \cap \mathrm{Sp}(4, \mathds{Z})\right) / \pm \mathbf{1}_4 \simeq \Gamma_0(P) \ltimes H_3(\mathds{Z}),
\end{equation}
where the embedding of $\Gamma_0(P),H_3(\IZ)$ in $\mathrm{Sp}(4,\IZ)$ is given by the maps $g_1,g_3$ (or $g^1,g^3$) respectively. 
\par

We now present the theorem we built towards in this subsection:
\begin{theorem}\label{thm:invariancegroup_PhiMN}
Let $G$ be the subgroup of $\rm{Sp}(4,\Z)$ under which $\Phi^{(M,N)}_k$ is a Siegel modular form of weight $k$. Then $G$ contains the following subgroups
\begin{align}\label{eq:inv_elements_threshold1}
 g_1(\Gamma_1(MN)),\qquad g^1(\Gamma(M,N)),\qquad g_4(\Gamma_1(M)),
\end{align} 
In addition, $G$ also contains the elements 
\begin{align}\label{eq:inv_elements_threshold2}
    g_3(\lambda M,\mu,\kappa M),\qquad \lambda,\mu,\kappa\in\Z,
    \\
    g^3(\lambda,\mu,\kappa),\qquad \lambda,\mu,\kappa\in\Z.
\end{align}
\end{theorem}

\begin{proof}
    The theorem is proved in appendix \ref{app:threshold_integral}. In the special case when $N=1$ or $M=1$, this theorem follows from \cite{jatkar2006dyon,david2006chl,david2006dyon,sen2008black,bhand2023mock}.
\end{proof}

\subsection{Fourier-Jacobi coefficients of the partition function}
We are now ready to describe the Fourier-Jacobi coefficients of the partition function. Let us consider the element 
\begin{equation}
    \begin{pmatrix}
 		1 & 0 & 1 & 1 \\ 0 & 1 & 1 & M \\ 0 & 0 & 1 & 0 \\ 0 & 0 & 0 & 1
 	\end{pmatrix}=g_1\big((\begin{smallmatrix}
 	1&1\\0&1    
 	\end{smallmatrix})\big)\cdot g_3(0,1,0)\cdot g_3(0,0, M)\in G.
\end{equation}
From this, we get periodicities of $\Phi_k^{(M,N)}$ to be
\begin{equation}
    \Phi_k^{(M,N)}(\tau+1,z,\sigma)=\Phi_k^{(M,N)}(\tau,z+1,\sigma)=\Phi_k^{(M,N)}(\tau,z,\sigma+M).
\end{equation}

The following two propositions form the basis of our main results:
\begin{prop}
    The partition function ${\Phi_k^{(M,N)}}^{-1}$ admits the Fourier-Jacobi expansions,
\begin{equation}\label{eq:Fourier_Jacobi_PhiMN}
\begin{split}
\frac{1}{\Phi_k^{(M,N)}(\tau,z,\sigma)} &=\sum_{m=-\alpha}^\infty\phi_m(\sigma,z)q^{m},\qquad \Img(\tau)\to\infty, \\ \frac{1}{\Phi_k^{(M,N)}(\tau,z,\sigma)} &=\sum_{n=-\gamma M}^\infty\psi_n(\tau,z)p^{n/M},\qquad \Img(\sigma)\to\infty,
\end{split}
\end{equation} for some functions $\phi_m$ and $\psi_n$.
\end{prop}
\begin{proof}
    The steps of the proof are identical to that in \cite[Proposition 5.1]{bhand2023mock}. Briefly, it follows from the fact that the function $q^\alpha {\Phi_k^{(M,N)}}^{-1}$ (respectively $p^\gamma {\Phi_k^{(M,N)}}^{-1}$) has poles in the $q$ (respectively $p$) plane at finite distance from the origin, seen from equation \eqref{eq:zeroes}, and it is single-valued as a function of $q$ (respectively $p^{1/M}$), thus admitting a Taylor series expansion as above.
\end{proof}

\begin{prop}
Let $\phi_m$, $\psi_n$ be defined as in equation \eqref{eq:Fourier_Jacobi_PhiMN}. Then,
\begin{enumerate}
    \item $\phi_m$ is a meromorphic Jacobi form of weight $-k$ and index $m$ for $\Gamma(M,N)\ltimes \Z^2$ with poles at $z\in\Z\sigma+\Z$, and 
    \item $\psi_n$ is a meromorphic Jacobi form of weight $-k$ and index $n/M$ for $\Gamma_1(MN)\ltimes(M\Z\times\Z)$ with poles at $z\in M\Z\tau+\Z$.
\end{enumerate}   
\end{prop}
\begin{proof}
    The proof is similar to that of \cite[Theorem 5.2]{bhand2023mock} and follows from the properties derived in theorem \ref{thm:invariancegroup_PhiMN} 
\end{proof}

We can now apply theorems \ref{dmzmaintheorem} and \ref{dmzmaintheorem2} to  decompose the meromorphic Jacobi forms $\phi_m$, $\psi_n$ into finite and polar parts and show that the completion of the finite parts are mock Jacobi forms.

In particular, the finite parts of $\phi_m,\psi_n$ is given by 
\begin{align}
\phi^F_m (\sigma,z) \coloneqq \sum \limits_{\ell \in \Z / 2m\Z} h^{\star, \phi}_{\ell} (\sigma) \, \vartheta_{m , \ell} (\sigma, z),
\\
\psi^F_n (\tau,z) \coloneqq \sum  \limits_{\ell \in \Z / 2n\Z} h^{\star, \psi}_{\ell} (\tau) \, \widetilde{\vartheta}_{\frac{n}{M} , \ell} (\tau, z),
\end{align}
where $h_\ell^{\star}$ is defined in \eqref{merthetacoeff}. To describe the polar parts explicitly, we expand $\phi_m$ and $\psi_n$ at $z=0$ as follows: 
\begin{align}
\phi_m\left(\sigma, \varepsilon\right)&=\frac{E^\phi_{0,m}(\sigma)}{(2 \pi i \varepsilon)^{2}}+\frac{D^\phi_{0,m}(\sigma)}{2 \pi i \varepsilon}+O(1) \quad  \text{as} \quad \varepsilon\to 0,
\\
\psi_n\left(\tau, \varepsilon\right)&=\frac{E^\psi_{0,n/M}(\tau)}{(2 \pi i \varepsilon)^{2}}+\frac{D^\psi_{0,n/M}(\tau)}{2 \pi i \varepsilon}+O(1) \quad \text{as} \quad \varepsilon\to 0.
\label{eq:}
\end{align}
From the behaviour at $z \to 0$ as in \eqref{eq:PhiMNatz=0}, we read off the residues. First, let us write,
\begin{equation}
\frac{1}{\Phi_k^{(M,N)}(\tau, z, \sigma)} \simeq-\frac{1}{4 \pi^2} \tilde{f}^{(k)}(\tau) \tilde{g}^{(k)}(\sigma) \frac{1}{z^2}+O\left(z^{-4}\right),    
\end{equation}
where $\tilde{f}^{(k)}, \tilde{g}^{(k)}$ are the inverses of the modular forms \eqref{eq:f_kg_kdef}:
\begin{equation}
\tilde{f}^{(k)}(\tau):=\frac{1}{f^{(k)}(\tau)}, \qquad \tilde{g}^{(k)}(\sigma):=\frac{1}{g^{(k)}(\sigma)} .
\end{equation}

Thus we have
\begin{equation}
E_{0,m}^\phi(\sigma)=\hat{f}^{(k)}_m \tilde{g}^{(k)}(\sigma), \qquad D_{0,m}^\phi(\sigma)=0
\end{equation}

where $\hat{f}^{(k)}_m$ is the coefficient of $q^m$ in the expansion of $\tilde{f}^{(k)}(\tau)$. Similarly,
\begin{equation}
E_{0,n / M}^\psi(\sigma)=\hat{g}^{(k)}_n \tilde{f}^{(k)}(\tau), \qquad D_{0,n / M}^\psi(\sigma)=0,
\end{equation}
where $\hat{g}^{(k)}_n$ is the coefficient of $p^{n / M}$ in the expansion of $\tilde{g}^{(k)}(\sigma)$. 

The polar parts are then computed to be given by
\begin{equation}
\begin{aligned}
& \phi_m^P(\sigma, z) \coloneqq \hat{f}^{(k)}_m \tilde{g}^{(k)}(\sigma) \sum_{s \in \mathds{Z}} \frac{p^{m s^2+s} y^{2 m s+1}}{\left(1-p^s y\right)^2}, \\
& \psi_n^P(\tau, z) \coloneqq \hat{g}^{(k)}_n \tilde{f}^{(k)}(\tau) \sum_{s \in \mathds{Z}} \frac{q^{M\left(n s^2+s\right)} y^{2 n s+1}}{\left(1-q^{M s} y\right)^2}.
\end{aligned}
\end{equation}

To summarize, it is shown by theorem \ref{dmzmaintheorem} that
\begin{equation}
    \phi_m = \phi^F_m + \phi^P_m, \qquad \psi_n = \psi^F_n + \psi^P_n
\end{equation} where $m$ and $n$ label the Fourier-Jacobi coefficients of ${\Phi_k^{(M,N)}}^{-1}$ as in equation (\ref{eq:Fourier_Jacobi_PhiMN}), and the finite and polar parts are described as above. Then, by theorem \ref{dmzmaintheorem2}, the finite parts of $\phi_m$ and $\psi_n$ are mixed mock Jacobi forms, and their completions, as described in \eqref{eq:finitepartcompletion}, are Jacobi forms of weight $-k$ and index $m$ and $n/M$ respectively, over the groups $\Gamma(M,N) \ltimes \Z^2$ and $\Gamma_1(MN) \ltimes (M\Z \times \Z)$. respectively

\subsection{Single-centered degeneracies}
We are now in a position to extract the degeneracies of the single-centered black holes. We begin with the following definition:
\begin{definition}
    We define the following sets of triplets
\begin{multline}
        A \coloneqq \bigg\{(n,\ell,m) \in \Z^3 \bigg| \frac{4nm}{M}-\ell^2 > 0, \\ \min_{r \in M\Z, \, s\in \Z, \, s(s+1) = 0 \, (\mathrm{mod} \, {r})}\bigg[\bigg(\frac{r\ell}{m} + 2s + 1\bigg)^2 + \frac{r^2}{m^2}\bigg(\frac{4nm}{M}-\ell^2 \bigg)\bigg] \geq 1 \bigg\}, \label{eq:setofcharges_1}
\end{multline} and
\begin{multline}
        B \coloneqq \bigg\{(n,\ell,m) \in \Z^3 \bigg| \frac{4nm}{M}-\ell^2 > 0, \\ \min_{r \in M\Z, \, s \in \Z, \, s(s+1) = 0 \, (\mathrm{mod} \, {r/M})}\bigg[\bigg(\frac{r\ell}{n} - 2s - 1\bigg)^2 + \frac{r^2}{n^2}\bigg(\frac{4nm}{M}-\ell^2 \bigg)\bigg] \geq 1 \bigg\}. \label{eq:setofcharges_2}
\end{multline}
\end{definition}

We now have all the ingredients to prove our main theorem.

\begin{theorem}\label{thm:mock_degen_AB} Let the black hole charges be parametrized as $P^2 = 2n$, $Q^2 = \dfrac{2m}{M}$ and $Q \cdot P = \ell$. Expand the mock Jacobi forms $\phi_m^F$ and $\psi_n^F$ as
\begin{equation}
    \phi_m^F(\sigma,z)=\sum_{n,\ell}c^F_{\phi_m}(n,\ell)p^{n/M}y^\ell,\qquad \psi^F_n(\tau,z)=\sum_{m,\ell}c_{\psi_n}^F(m,\ell)q^my^\ell.
\end{equation}
Then,
\begin{enumerate}[(a)]
    \item For $(n,\ell, m) \in A$, the single-centered black hole degeneracy $d(\vec{Q},\vec{P})$ is given by the Fourier coefficient of the mock Jacobi form $\phi^F_m$ as
    \begin{equation}
    	d(\vec{Q},\vec{P}) = \frac{(-1)^{\ell+1}}{M}c_{\phi_m}^F(n,\ell).
    \end{equation}
    \item For $(n,\ell, m) \in B$, the single-centered black hole degeneracy $d(\vec{Q},\vec{P})$ is given by the Fourier coefficient of the mock Jacobi form $\psi^F_n$ as
    \begin{equation}
    	d(\vec{Q},\vec{P}) = \frac{(-1)^{\ell+1}}{M}c_{\psi_n}^F(m,\ell).
    \end{equation}
\end{enumerate}
\end{theorem}

\begin{proof}
    We follow the proof of \cite[Theorem 5.4]{bhand2023mock}. In order to find the single-centered black hole degeneracies, we must perform the contour integral (\ref{degeneracyformula}). To do this while using the Fourier-Jacobi expansion (\ref{eq:Fourier_Jacobi_PhiMN}), we must deform the contour $\mathcal{C}$ to an identical contour $\mathcal{C}'$ but at $\Img(\tau) \to \infty$. It is shown in the appendix \ref{prop:deformationofcontour} that for the set of charges $A$ (defined in (\ref{eq:setofcharges_1} and \ref{eq:setofcharges_2})), the partition function ${\Phi^{(M,N)}_k}^{-1}$ does not have any poles in the region bounded by $\mathcal{C}$ and $\mathcal{C}'$. Thus, we restrict our analysis to this set of charges and proceed. We have,
\begin{align}
	d(n,\ell,m) &= \frac{1}{M}(-1)^{\ell + 1} \int_{\mathcal{C}'} d\tau dz d\sigma \, q^{-m} \, y^{-\ell} \, p^{-n/M} \sum_{a=-\alpha}^\infty\phi_a(\sigma,z)q^a, \\
    &= \frac{1}{M}(-1)^{\ell + 1} \int_{\mathcal{C}(n,\ell)} dz d\sigma \, y^{-\ell} \, p^{-n/M} \phi_m(\sigma,z),
\end{align} where $\mathcal{C}(n,\ell)$ is the projection of $\mathcal{C}$ onto the $\sigma-z$ hypersurface. Next, we perform the $z$-integral. Note that for the attractor contour (\ref{eq:attractorcontour}),
\begin{equation}
    \Img(z) = -\frac{\ell}{2m} \Img(\sigma).
\end{equation} This motivates shifting the integral over $z$ by a constant,
\begin{equation}
    -\frac{\ell}{2m} \Rel(\sigma) \leq \Rel(z) \leq -\frac{\ell}{2m} \Rel(\sigma) + 1.
\end{equation} This is valid because we do not cross any pole in doing so. This simplifies the contour integral as,
\begin{align}
	d(n,\ell,m) &= \frac{1}{M}(-1)^{\ell + 1} \int_{\mathcal{R}_N} d\sigma \, p^{-n/M} \int_{-\frac{\ell}{2m} \sigma}^{-\frac{\ell}{2m} \sigma + 1} dz \, y^{-\ell} \phi_m(\sigma,z), \\
    &= \frac{1}{M}(-1)^{\ell + 1} \int_{\mathcal{R}_N} d\sigma \, p^{-n/M + \ell^2/4m} h_\ell^\star (\sigma),
\end{align} where $\mathcal{R}_t$ is any horizontal line of length $t$ in the corresponding complex plane, and $h_\ell^\star$ is defined in (\ref{merthetacoeff}).

Now recall the definition of the Fourier part of the Jacobi form $\phi_m$,
\begin{align}
    \phi^F_m (\sigma, z) &= \sum_{\ell = 0}^{2m - 1}  h_\ell^\star (\sigma) \vartheta_{m,\ell} (\sigma, z) \\
    &= \sum_{\ell = 0}^{2m - 1} \sum_{r \in \Z} p^{(\ell + 2mr)^2/4m} y^{\ell + 2mr} h_\ell^\star (\sigma) \\
    &= \sum_{\ell \in \Z} p^{\ell^2/4m} y^\ell h_\ell^\star (\sigma).
\end{align}

Inserting this in the contour integral we get,
\begin{align}
	d(n,\ell,m) &= \frac{1}{M}(-1)^{\ell + 1} \int_{\mathcal{R}_N} d\sigma \int_{\mathcal{R}_1} dz \, p^{-n/M} y^{-\ell} \phi^F_m (\sigma, z).
\end{align}

This proves part (a) of the theorem. Part (b) is proven analogously. This time we deform the attractor contour $\mathcal{C}$ to an identical one $\mathcal{C}''$ except that $\Img(\sigma) \to \infty$. It is shown in appendix \ref{prop:deformationofcontour2} that for the set of charges $B$ (defined in \eqref{eq:setofcharges_1} and \eqref{eq:setofcharges_2}), the partition function does not have any poles in the region bounded by $\mathcal{C}$ and $\mathcal{C}''$, allowing us to perform the deformation without picking up any residues. Then, following the same steps as before, we get,
\begin{align}
	d(n,\ell,m) &= \frac{1}{M}(-1)^{\ell + 1} \int_{\mathcal{R}_1} d\tau \, q^{-m + \ell^2M/4n} h_\ell^\star (\tau).
\end{align}

Again, we find the Fourier part of the Jacobi form $\psi_n$,
\begin{align}
    \psi^F_n (\tau, z) &= \sum_{\ell = 0}^{2n - 1}  h_\ell^\star (\tau) \widetilde{\vartheta}_{\frac{n}{M},\ell} (\tau, z) \\
    &= \sum_{\ell = 0}^{2n - 1} \sum_{r \in \Z} q^{(\ell + 2nr)^2 M/4n} y^{\ell + 2nr} h_\ell^\star (\tau) \\
    &= \sum_{\ell \in \Z} q^{\ell^2 M/4n} y^\ell h_\ell^\star (\tau).
\end{align}

Inserting this in the contour integral we get,
\begin{align}
	d(n,\ell,m) &= \frac{1}{M}(-1)^{\ell + 1} \int_{\mathcal{R}_1} d\tau \int_{\mathcal{R}_1} dz \, q^{-n} y^{-\ell} \psi^F_m (\tau, z).
\end{align}
\end{proof}

As in the CHL case \cite{bhand2023mock}, we find that there are some charges whose degeneracies cannot be directly related to mock Jacobi forms \textit{without using $S$-duality}. This is due to the wall crossing behaviour of the degeneracies across the walls of marginal stability. More precisely, suppose for a charge $(n,\ell,m)$, we encounter a pole while deforming the contour $\CC$ to either $\CC'$ or $\CC''$, then we pick up the residue at the pole and the degeneracy of the black holes with charge $P^2=2n,Q^2=2m/M,P\cdot Q=\ell$ is not correctly captured by the corresponding coefficient of the mock Jacobi forms appearing in theorem \ref{thm:mock_degen_AB}. The polar part of the Siegel modular form cannot remove the singular contributions for these special charges. As shown in \cite{bhand2023mock}, the residue at the pole is related to the wall-crossing behaviour via the prescription given in \cite{Mandal:2010cj,sen2008black}. These charge sectors need further analysis and we postpone that study for the future.  In particular, we would like to understand whether a charge which is not in the sets $A$ or $B$ can be mapped to a charge in $A$ or $B$ using $S$-duality. 
  
\section{Dyon partition function through lifts}\label{sec:4}
In this section, we provide a multiplicative lift construction of the partition function for $\CM=\text{K}3$ and comment on the existence of an additive lift construction. The main tool is the notion of multiplicative lifts constructed in \cite{Persson:2013xpa}. Toward the end, we present a comparison with the additive lift partition function for the system for specific values of $(M,N)$ in \cite{govindarajan2011bkm}.

\subsection{Multiplicative and additive lifts}
In this section, we will review the construction of multiplicative lift as presented in \cite{Persson:2013xpa} in the simplified setup relevant for our purposes. The study requires implementation of the Mathieu Moonshine theorem, as we discuss below.

\subsubsection{Twisted-twined elliptic genus and Mathieu Moonshine}
Let us begin with the definition of \textit{Ramond representation}: a representation of the superconformal algebra is called Ramond if the action of the supercurrent is integer modded. Mathieu Moonshine theorem \cite{Eguchi:2010ej,Eguchi:2010fg,Gaberdiel:2010ca,Gaberdiel:2010ch,Cheng:2010pq,Gannon:2012ck} states that there exists a unitary Ramond representation $\mathcal{H}$ of the $\mathcal{N}=4$ superconformal algebra at central charge $c=6, \mathds{Z}_2$-graded by the right-moving fermion number $(-1)^{\bar{F}}$, that carries a non-trivial action of the Mathieu group $\mathds{M}_{24}$ commuting with both the $\mathcal{N}=4$ algebra and the $\mathds{Z}_2$-grading. Then for each $g \in \IM_{24}$, the twined elliptic genus
defined by
\begin{equation}
\phi_g(\tau, z):=\mathrm{Tr}_{\mathcal{H}}\left(g q^{L_0-\frac{c}{24}} y^{J_L}(-1)^{F+\bar{F}}\right), \qquad g \in \IM_{24},
\end{equation}
is a weak Jacobi form of weight 0 and index 1 (possibly with a multiplier) under $\Gamma_0(N)$, where $N$ is the order of $g$. In particular, $\phi_e$ is the elliptic genus of K3. 

The generalized Mathieu Moonshine conjecture was proposed in \cite{Gaberdiel:2012gf}. It states that for every $g\in\IM_{24}$, there exists a $\IZ_2$-graded unitary Ramond representation $\mathcal{H}_g$ of the $\mathcal{N}=4$ superconformal algebra at central charge $c=6$, which is also a projective representation\footnote{A complex representation $(\rho,V)$ of a group $G$ is called projective if $\rho(g_1g_2)=e^{i\alpha(g_1,g_2)}\rho(g_1)\rho(g_2)$ for some function $\alpha:G\times G\to \IR$.} $\rho_g: C_{\IM_{24}}(g) \longrightarrow \mathrm{GL}\left(\mathcal{H}_g\right)$ of the centralizer\footnote{The centralizer of a subset $S$ of a group $G$ is defined as $C_G(S)=\{g \in G:g s=s g$ for all $s \in S\}$.} of $g$ in $\IM_{24}$. The operators $\rho_g(h)$ commute with the $\mathcal{N}=4$ algebra and with the right-moving fermion number $(-1)^{\bar{F}}$. The twisted-twined elliptic genus 
defined for each pair of commuting elements $g, h \in \IM_{24}$ by
\begin{equation}\label{eq:phi_gh}
\phi_{g, h}(\tau, z):=\operatorname{Tr}_{\mathcal{H}_g}\left(\rho_g(h) q^{L_0-\frac{c}{24}} y^{J_L}(-1)^{F+\bar{F}}\right),
\end{equation}
is an elliptic function. With respect to $\mathrm{SL}(2,\IZ)\times\IM_{24}$, it is invariant under the $\alpha$-twisted slash operator, where $\alpha\in H^3(\IM_{24},\mathrm{U}(1))$ is the group cohomology class which classifies the central extension of $\IM_{24}$ by U(1) corresponding to the projective representation $\rho_g$. We refer the readers to \cite{Gaberdiel:2012gf,Persson:2013xpa} for the precise definition. Strong evidence was given for the generalized Mathieu Moonshine conjecture but it is not proved. The spaces $\CH,\CH_{g}$ in these Moonshine phenomena have not been explicitly constructed and is a wide open problem.

Since the original Mathieu Moonshine phenomena was observed in the elliptic genus of K3 surface \cite{Eguchi:2010ej}, one would expect that the (twisted) Ramond sector Hilbert space of the SCFT with target space K3 might be a good candidate for the spaces $(\CH_g)~\CH$. But it was observed in \cite{Gaberdiel:2011fg,Gaberdiel:2012um} that no known SCFT with target space a K3 surface has $\IM_{24}$ as the symmetry group. For our purposes, this problem does not arise as we now explain. Any symplectic automorphism of K3 surface automatically induces a symmety of the SCFT. By Mukai's theorem \cite{Mukai1988} any subgroup of $\IM_{23}$ can be realized as a subgroup of symplectic automorphisms of a K3 surface. Thus, for any $g\in\IM_{23}$, the space $\CH_{g}$ is the RR sector of a K3 sigma model. Thus, for our cases, the twisted-twined elliptic genus can be explicitly computed and we present the same in the next subsection. 

\subsubsection{Construction of the lifts}
Let us consider the twisted-twined elliptic genus $\phi_{g,h}$ for a pair of commuting elements $g,h\in\IM_{23}$. Expanding $\phi_{g, h}(\tau, z)$ as\footnote{Note that $\lambda=1$ compared to eq. (3.44) of \cite{Persson:2013xpa} since the multiplier is trivial for $g\in\IM_{23}$, see appendix D.2.1 of \cite{Persson:2013xpa}.}
\begin{equation}
\phi_{g, h}(\tau, z)=\sum_{n \in \mathbb{Z} / M} \sum_{\ell \in \mathbb{Z}} c_{g, h}(n, \ell) q^n y^l.
\end{equation}
Here we define\footnote{Compared to eq. (3.49) of \cite{Persson:2013xpa}, we have changed variable $b \rightarrow-b$ for ease of comparison.}
\begin{equation}
\hat{c}_{g, h}(d, m, \ell, t):=\frac{1}{M N} \sum_{k=0}^{N-1} \sum_{b=0}^{M-1} \bold{e}^{-\frac{t k}{N}} \bold{e}^{-\frac{b m}{M}} c_{g^d, g^{b} h^k}(m d / M, \ell).
\end{equation}
Then the multiplicative lift of $\phi_{g, h}$ is defined by\footnote{\label{foot:sig_tau_ex}Compared to eq. (3.51) of \cite{Persson:2013xpa} we have changed $\sigma \leftrightarrow \tau$ for ease of comparison.}:
\begin{equation}\label{eq:mult_lift_gh}
\Phi_{g, h}(\tau, z, \sigma)=\mathrm{Mult}[\phi_{g,h}]:=q p^{\frac{1}{M}} y \prod_{(d, m, \ell)>0} \prod_{t=0}^{N-1}\left(1-\bold{e}^{\frac{t}{N}} q^d y^\ell p^{\frac{m}{M}}\right)^{\hat{c}_{g, h}(d, m, \ell, t)},
\end{equation}
where $(d, m, \ell)>0$ means
\begin{equation}
\begin{aligned}
& d, m \in \Z_{\geq 0} \text { and }\left\{\begin{array}{l}
\ell \in \Z, \qquad \ell<0 \qquad \text { if } m=d=0, \\
\ell \in \Z, \qquad \text { otherwise. }
\end{array}\right.
\end{aligned}
\end{equation}
$\phi_{g,h}$ is called the multiplicative seed function of $\Phi_{g,h}$. 

Let us also define the \textit{second quantized twisted-twined elliptic genus} \cite{Persson:2013xpa}. For this we first need to define the action of \textit{equivariant Hecke operators}\footnote{Equivariant Hecke operators were first defined by Carnahan to prove the generalized Moonshine conjecture \cite{Carnahan_2010,Carnahan_2012}. They were also used to prove infinite product formulas for certain mock modular forms \cite{bhand2025generalisation}.} on $\phi_{g, h}$: for each $L \in \N$ define $\CT_L \phi_{g, h}$ by

\begin{equation}
(\CT_L \phi_{g, h})(\tau, z):=\frac{1}{L} \sum_{a d=L} \sum_{b=0}^{d-1} \phi_{g^d, g^{-b} h^a}\left(\frac{a \tau+b}{d}, a z\right).
\end{equation}
Then the second quantized twisted-twined elliptic genus is defined by
\begin{equation}
\Psi_{g, h}(\tau, z, \sigma):=\exp \left(\sum_{L=1}^{\infty} q^L\left(T_L \phi_{g, h}\right)(\sigma, z)\right).
\end{equation}
Then a tedious but straightforward calculation shows that \cite{Persson:2013xpa}
\begin{equation}
\Phi_{g, h}(\tau, z, \sigma)=\frac{q \, \psi_{g, h}(\sigma, z)}{\Psi_{g, h}(\tau, z, \sigma)}.
\end{equation}
where
\begin{equation}
\begin{aligned}\label{eq:add_seed_gh}
& \psi_{g, h}(\sigma, z):=p^{\frac{1}{M}} y \prod_{t=0}^{N-1}\left(\prod_{\ell<0}\left(1-\bold{e}^{\frac{t}{N}} y^\ell\right)^{\hat{c}_{g, h}(0,0, \ell, t)}\right)\left(\prod_{\ell \in \Z} \prod_{m=1}^{\infty}\left(1-\bold{e}^{\frac{t}{N}} p^{\frac{m}{M}} y^\ell\right)^{\left.\hat{c}_{g, h}(0, m, \ell, t\right)}\right).
\end{aligned}
\end{equation}
It can be verified that $\Phi_{g, h}(\tau, z, \sigma)$ satisfies the $S$-duality transformation
\begin{equation}
\Phi_{g, h}\left(\frac{\sigma}{M}, z, M \tau\right)=\Phi_{g, h^{\prime}}(\tau, z, \sigma),
\end{equation}
where $h'$ is some ``relabelled'' element of $\IM_{24}$ not necessarily in the same conjugacy class as $h$. In the next section, we will see examples where $h=h'$.

Now observe that given a Jacobi form $\psi$ of weight $k$, index $t$, with character $\chi$ with conductor\footnote{Recall that the conductor of a Dirichlet character $\chi$ modulo $k$ is the smallest positive integer $m_\chi$ such that we can write $\chi=\chi_1\cdot\chi_2$ where $\chi_1$ is a primitive character modulo $m_\chi$.} $m_\chi$, its additive lift is defined as follows \cite{Gritsenko:1996ax,Gritsenko:1996tm}:
\begin{equation}
\Phi:=\operatorname{Add}[\psi]=\sum_{\substack{m \geq 1\\m\equiv 1\bmod m_\chi}} m^{2-k}\left(T_m^{-} \widetilde{\psi}\right)(\tau, z, \sigma),
\end{equation}
where
\begin{enumerate}
    \item $\tilde{\psi}(\tau, z, \sigma)=q \psi(\sigma, z)$.
\item $T_m^{-}$ are certain Hecke operators defined as follows:
\begin{equation}
\left(T_m^{-} \cdot \tilde{\psi}\right)(\tau, z, \sigma):=m^{2 k-3} \sum_{\substack{a d=m \\ b \bmod d}} \chi(\sigma_a)d^{-k} \psi\left(\frac{a \sigma+bm_\chi}{d}, a z\right) q^{m t},
\end{equation}
where $\sigma_a=(\begin{smallmatrix}
    a^{-1}&0\\0&a
\end{smallmatrix})$.
\end{enumerate}
$\psi$ is called the \textit{additive seed function} of $\Phi$. 

As explained in \cite{Persson:2013xpa}, if there exists an additive lift description of $\Phi_{g, h}$, then the seed function will be $\psi_{g, h}$ and hence
\begin{equation}
\Phi_{g, h}=\mathrm{Mult}\left[\phi_{g, h}\right]=\operatorname{Add}\left[\psi_{g, h}\right].
\end{equation}

\subsection{Partition function from multiplicative lift}\label{sec:mult_lift_part_fu}
Let us now compare the partition function \eqref{eq:product_PhiMN} with the multiplicative lift \eqref{eq:mult_lift_gh}. We begin by noticing that since $\phi_{g, h}$ is a Jacobi form, we have the expansion
\begin{equation}
\phi_{g, h}(\tau, z)=\sum_{n \in \Z / M} \sum_{\ell \in \Z} c_{g, h}(4 n-\ell^2) q^n y^{\ell}.
\end{equation}
Taking $g=g_M^{r^{\prime}} g_N^r$, and $h=g_M^{s^{\prime}} g_N^s$ and comparing \eqref{eq:phi_gh} with \eqref{eq:cbrsr's'def} we see that
\begin{equation}
c_{g,h}(n,\ell)=M N \cdot \begin{cases}c_0^{(r, s ; r^{\prime}, s^{\prime})}(4 n-\ell^2) ; & \ell \text { even }, \\ c_1^{(r, s ; r^{\prime}, s^{\prime})}(4 n-\ell^2) ; & \ell \text { odd .}\end{cases}
\end{equation}
Thus for $\ell \in 2 \mathbb{Z}+b$, we have
\begin{equation}
\begin{aligned}
\hat{c}_{g_M, g_N}(d, m, \ell, t) & =\frac{1}{M N} \sum_{s=0}^{N-1} \sum_{s^{\prime}=0}^{M-1} \bold{e}^{-\frac{t s}{N}} \bold{e}^{-\frac{s^{\prime} m}{M}} c_{g_M^d, g_M^{s^{\prime}} g_N^s}(m d / M, \ell) \\
& =\sum_{s=0}^{N-1} \sum_{s^{\prime}=0}^{M-1} \bold{e}^{-\frac{t s}{N}} \bold{e}^{-\frac{s^{\prime} m}{M}} c_b^{(0, s ; d, s^{\prime})}(4m d / M-\ell^2).
\end{aligned}
\end{equation}
Next, for $\CM=\text{K}3$, the exponents $\alpha,\gamma$ in \eqref{eq:alpha_def} and \eqref{eq:gamma_def} take the special values \cite{Chattopadhyaya:2017ews}, \cite[Eq. (2.12)]{sen2010discrete} 
\begin{equation}
    \alpha=1,\qquad \gamma=\frac{1}{M}.
\end{equation}
Noting that in the product \eqref{eq:product_PhiMN}, $k'M\equiv r'\bmod M$ and the order of $g_M$ is $M$,
we get
\begin{equation}
\begin{aligned}
\Phi_k^{(M, N)}(\tau, z, \sigma) & =\Phi_{g_M, g_N}\left(\frac{\sigma}{M}, z, M \tau\right) \\
& \overset{!}{=}\Phi_{g_M, g_N}(\tau, z, \sigma),
\end{aligned}
\end{equation}
where the last equality does not always hold. For example, it holds only when the ``relabelled'' element $g_N'=g_N$. 
Thus, we have represented the partition function as a multiplicative lift supporting the conjecture in \cite{Persson:2013xpa} for some special cases. 

\subsection{Partition function from additive lift}
As discussed above, it is not always guaranteed that the dyon partition function will admit an additive lift description. In this section, we collect some examples from \cite{Persson:2013xpa} when the dyon partition function can be expressed as an additive lift. 

\paragraph{$\boldsymbol{(M,N)=(2,2).}$} This is the case when $g_M\in 2A$ and $g_N\in 2A_2,2A_3,2A_5$, the dyon partition function is identical \cite{Persson:2013xpa} and can be identified with one of the Siegel modular forms from \cite{Gritsenko:1996tm}:
\begin{equation}
    \Phi_{2 A, 2 A_2}(\tau, z,\sigma)=\Phi_{2 A, 2 A_3}(\tau, z,\sigma)=\Phi_{2 A, 2 A_5}\left(\tau, z,\sigma\right)=[\Delta_2(\tau, z,\sigma)]^2.
\end{equation}
These functions also satisfy
\begin{equation}
    \Phi_{2 A, 2 A_2}(\tau, z,\sigma)=\Phi_{2 A, 2 A_3}\left(\frac{\sigma}{2},z, 2\tau\right), \qquad \Phi_{2 A, 2 A_5}(\tau, z,\sigma)=\Phi_{2 A, 2 A_5}\left(\frac{\sigma}{2},z, 2\tau\right).
\end{equation}
The dyon partition function in this case is a Siegel modular form of weight 4 for the paramodular group\footnote{The conjugation by $V_1$ comes from the fact that the dyon partition function in this paper is related to the multiplicative lift in \cite{Persson:2013xpa} by $\sigma\leftrightarrow\tau$, see Footnote \ref{foot:sig_tau_ex}.} $V_1\Gamma_2^+(1)V_1$ with a nontrivial multiplier $\nu_{g_M,g_N}$ specified by its value on the elements\footnote{Here we used the fact that $V_1g_1V_1=g^1$.} $g^1(0,1,-1,0)$, $g^1(1,1,0,1)$, $V_1V_2V_1$, $V_1 \, g_3(\lambda,\mu,\kappa) \, V_1$, $V_1 \, g_3(0,0,\kappa/2) \, V_1$, $\lambda,\mu,\kappa\in \Z$:
\begin{equation}\label{eq:multiplier_def}
\begin{split}
\nu_{g_M,g_N}(g^1(0,1,-1,0))=\nu_{g_M,g_N}(g^1(1,1,0,1))=-1,\qquad \nu_{g_M,g_N}(V_1V_2V_1)=1,\\ \nu_{g_M,g_N}(V_1 \, g_3(\lambda,\mu,\kappa) \, V_1) =1,\qquad \nu_{g_M,g_N}(V_1 \, g_3(0,0,\kappa/2) \, V_1) = \bold{e}^{\kappa/2}.
\end{split}
\end{equation}
Let us check that the multiplier is trivial on elements of the form \eqref{eq:inv_elements_threshold1} and \eqref{eq:inv_elements_threshold2}. $\Gamma_0(2)$, which is generated by 
\begin{equation}
    T:=\begin{pmatrix}
        1&1\\0&1
    \end{pmatrix},\qquad R:=ST^{-2}S=\begin{pmatrix}
        1&0\\2&1
    \end{pmatrix},
\end{equation}
where
\begin{equation}
    S:=\begin{pmatrix}
        0&-1\\1&0
    \end{pmatrix},
\end{equation}
contains $\Gamma_1(2)$. It is easy to see that 
\begin{equation}
    \Gamma^1(P)=S\Gamma_1(P)S^{-1}.
\end{equation}
As 
\begin{equation}
\nu_{g_M,g_N}(g^1(T^2))=\nu_{g_M,g_N}(g^1(ST^{-2}S))=1,    
\end{equation}
$\nu_{g_M,g_N}$ is trivial on $\Gamma(2,2)=\Gamma_1(2)\cap \Gamma^1(2)$. 

Next, note that 
\begin{equation}\label{eq:g1_g1_conjugation}
    g_1\begin{pmatrix}
    a&b\\cMN&d    
    \end{pmatrix}=(V_1V_MV_1) \, g^1\begin{pmatrix}
        a&bM\\cN&d
    \end{pmatrix}(V_1V_MV_1).
\end{equation}
Thus, we have 
\begin{align}
\nu_{g_M,g_N}|_{g_1(\Gamma_1(4))} &= \nu_{g_M,g_N}(V_1V_2V_1) \, \nu_{g_M,g_N}|_{g^1(\Gamma(2,2))} \,\nu_{g_M,g_N}(V_1V_2V_1) \\ &= \nu_{g_M,g_N}|_{g^1(\Gamma(2,2))}=1. 
\end{align}
Next, we note that
\begin{multline}
    V_1 \big(U_M \,g_1(a,b,cM,d) \, g^1(a,-b,-cM,d) \, U_M^{-1} \big)V_1 =\begin{pmatrix}
        d &-c&\frac{c}{M}&0
        \\
        -bM&a&0&b
        \\
        0&0&a&bM
        \\
        0&0&c&d
    \end{pmatrix}
    \\=g_4(d,-c,-bM,a) \, g_3(0,bc,bd) \, V_1 \, g_3(0,0,ac/M) \, V_1.  
\end{multline}
For 
\begin{equation}
    \gamma=\begin{pmatrix}
        a&b\\cM&d
    \end{pmatrix}\in\Gamma_1(M),\qquad \tilde{\gamma}=\begin{pmatrix}
        a&-b\\-cM&d
    \end{pmatrix}\in\Gamma_1(M),
\end{equation}
we have
\begin{equation}
\begin{split}
\nu_{g_M,g_N}(V_1(U_Mg_1(\gamma)g^1(\tilde{\gamma})U_M^{-1})V_1) &=\nu_{g_M,g_N}(V_1U_MV_1) \, \nu_{g_M,g_N}(V_1g_1(\gamma)g^1(\tilde{\gamma})V_1) \, \nu_{g_M,g_N}(V_1U_M^{-1}V_1)
\\&=\nu_{g_M,g_N}(V_1g_1(\gamma)V_1) \, \nu_{g_M,g_N}(V_1g^1(\tilde{\gamma})V_1)\\&=\nu_{g_M,g_N}(g^1(\gamma)) \, \nu_{g_M,g_N}(g_1(\tilde{\gamma})),
\end{split}    
\end{equation} where we used the fact that $V_1g^1V_1=g_1$. Next, note that 
\begin{equation}
    \nu_{g_M,g_N}(g_1(\tilde{\gamma}))=\nu_{g_M,g_N}(g_1(\gamma))^{-1}.
\end{equation}
This follows from 
\begin{equation}
    \tilde{S}=S^{-1},\qquad \tilde{T}=T^{-1},
\end{equation}
and 
\begin{equation}
    \nu_{g_M,g_N}(AB)=\nu_{g_M,g_N}(A) \, \nu_{g_M,g_N}(B)=\nu_{g_M,g_N}(BA).
\end{equation}
This implies 
\begin{equation}
\nu_{g_M,g_N}\big(V_1(U_Mg_1(\gamma)g^1(\tilde{\gamma})U_M^{-1})V_1\big)=\nu_{g_M,g_N}(g^1(\gamma)g_1(\gamma)^{-1}).    
\end{equation}
Thus, we have 
\begin{equation}
\begin{split}
\nu_{g_M,g_N}(g_4(d,-c,-bM,a))&=\nu_{g_M,g_N}(g^1(\gamma)g_1(\gamma)^{-1}) \, \nu_{g_M,g_N}(V_1 \, g_3(0,0,ac/M)V_1)^{-1} \, \nu_{g_M,g_N}(g_3(0,bc,bd))^{-1}
\\&=\bold{e}^{-ac/M} \, \nu_{g_M,g_N}(g_3(0,bc,0))^{-1} \, \nu_{g_M,g_N}(g^1(1,bd,0,1))^{-1} \, \nu_{g_M,g_N}(g^1(\gamma)g_1(\gamma)^{-1})
\\
&=\bold{e}^{-ac/M+bd/2} \, \nu_{g_M,g_N}(g^1(\gamma)g_1(\gamma)^{-1}).
\end{split}
\end{equation}
We now check that this phase is trivial on $\Gamma_0(2)$ which contains $\Gamma_1(2)$. We have 
\begin{equation}
\begin{split}
\nu_{g_M,g_N}(T))&=e^{i\pi}\nu_{g_M,g_N}(g^1(ST^2S) \, g_1(ST^2S)^{-1})
\\
&=-\nu_{g_M,g_N}(g^1(STS)^{-1})=1,
\end{split}
\end{equation}
where we used \eqref{eq:g1_g1_conjugation} with $N=1$ in the second step. Similarly,
\begin{equation}
\begin{split}
\nu_{g_M,g_N}(ST^{-2}S))&=e^{-i\pi}\nu_{g_M,g_N}(g^1(T^{-1}) \,g_1(T^{-1})^{-1})
\\
&=\nu_{g_M,g_N}(g^1(T^2))=1.    
\end{split}    
\end{equation}
Finally, for $\lambda,\mu,\kappa\in\Z$ we have 
\begin{equation}
\nu_{g_M,g_N}(g^3(\lambda ,\mu,\kappa ))=\nu_{g_M,g_N}(V_1 \, g_3(\lambda,\mu,\kappa ) V_1)=1,    
\end{equation}
and 
\begin{equation}\label{eq:trivial_mult_final}
\begin{split}
    \nu_{g_M,g_N}(g_3(\lambda M,\mu,\kappa M))&=\nu_{g_M,g_N}(V_1V_2V_1 \, g^3(\lambda ,\mu,\kappa )V_1V_2V_1)\\&=\nu_{g_M,g_N}(V_1V_2V_1) \, \nu_{g_M,g_N}(g^3(\lambda ,\mu,\kappa )) \, \nu_{g_M,g_N}(V_1V_2V_1)\\&=1.  
\end{split}  
\end{equation}
From \cite{Gritsenko:1996tm}, we know that there is an additive lift description of these Siegel modular forms. The seed function in this case is given by \cite{Persson:2013xpa}
\begin{equation}
    \psi_{g_M, g_N}(\tau, z)=-\vartheta_1(\tau, z)^2 \eta(\tau)^6,
\end{equation}
which agrees with the seed in \cite{govindarajan2011bkm} up to a minus sign. 
\paragraph{$\boldsymbol{(M,N)=(2,4)}$.} This is the case when we take $g_M\in 2A,g_N\in 2B_1,2B_2$. The multiplicative lift in these two cases are identical \cite{Persson:2013xpa} and are related to Siegel modular forms constructed by Clery and Gritsenko \cite{clery2011siegel}:
\begin{equation}
    \Phi_{2 A, 2 B_1}(\tau, z,\sigma)=\Phi_{2 A, 2 B_2}(\tau, z,\sigma)=Q_1(\tau, z,\sigma)^2.
\end{equation}
These functions also satisfy 
\begin{equation}
    \Phi_{2 A, 2 B_{1,2}}\left(V_2 \cdot \Omega\right)=\Phi_{2 A, 2 B_{1,2}}(\Omega) .
\end{equation}
The dyon partition function in this case is a Siegel modular form of weight 2 for the paramodular group $\Gamma_2^+(2)$ with multiplier $\nu_{g_M,g_N}$ given by
\begin{equation}
\begin{split}
\nu_{g_M,g_N}(g^1(1,0,2,1)) =\nu_{g_M,g_N}(g^1(1,1,0,1)) &=-1, \\ \nu_{g_M,g_N}(V_1V_2V_1) &=1,\\ \nu_{g_M,g_N}(g^3(\lambda,\mu,\kappa)) &=1,\\ \nu_{g_M,g_N}(g^3(0,0,\kappa/2)) &= \bold{e}^{\kappa/2}.
\end{split}
\end{equation}
Again one can show that the character is trivial on elements of the form \eqref{eq:inv_elements_threshold1} and \eqref{eq:inv_elements_threshold2} using similar analysis as in \eqref{eq:multiplier_def} to \eqref{eq:trivial_mult_final}.
\par
By \cite{clery2011siegel}, the dyon partition function admits additive lift description and the additive seed is given by 
\begin{equation}
    \psi_{g_M, g_N}(\tau, z)=-\frac{\vartheta_1(\tau, z)^2}{\eta(\tau)^2} \eta(2 \tau)^4,
\end{equation}
which again agrees with the seed in \cite{govindarajan2011bkm} up to a negative sign. 
\paragraph{$\boldsymbol{(M,N)=(4,4)}$.} This is the case when we take $g_M\in 4B,g_N\in 4B_4,4B_7$. The multiplicative lifts are identical \cite{Persson:2013xpa} in these two cases and can be identified with a Siegel modular form constructed by Gritsenko and Nikulin \cite{Gritsenko:1996tm}:
\begin{equation}
    \Phi_{4 B, 4 B_4}\left(\Omega\right)=\Phi_{4 B, 4 B_7}(\Omega)=\Delta_{1/2}(\Omega)^2 .
\end{equation}
It also satisfies 
\begin{equation}
    \Phi_{4 B, 4 B_4}\left(V_4 \cdot \Omega\right)=\Phi_{4 B, 4 B_7}(\Omega) .
\end{equation}
The dyon partition function is thus a Siegel modular form of weight 1 for $\Gamma_4^+(1)$ with multiplier $\nu_{g_M,g_N}$ given by
\begin{equation}
\begin{split}
\nu_{g_M,g_N}(g^1(0,-1,1,0))=\nu_{g_M,g_N}(g^1(1,1,0,1)) &=-i,\\ \nu_{g_M,g_N}(V_4) &=1,\\ \nu_{g_M,g_N}(g^3(\lambda,\mu,\kappa)) &=1, \\ \nu_{g_M,g_N}(g^3(0,0,\kappa/4)) &= \bold{e}^{\kappa/4}.
\end{split}
\end{equation}
Again one can show that the character is trivial on elements of the form \eqref{eq:inv_elements_threshold1} and \eqref{eq:inv_elements_threshold2} using similar analysis as in \eqref{eq:multiplier_def} -- \eqref{eq:trivial_mult_final}.
\par
From \cite{Gritsenko:1996tm}, the dyon partition function admits an additive lift description and the seed function is given by \cite{Persson:2013xpa}
\begin{equation}
    \psi_{g_M, g_N}(\tau, z)=-\vartheta_1(\tau, z)^2,
\end{equation}
which again matches with the seed function in \cite{govindarajan2011bkm} up to a negative sign.
\paragraph{$\boldsymbol{(M,N)=(3,3)}$.} This is the case when $g_M\in 3A,g_N\in 3A_3$. The multiplicative lift in this case can be identified with the square of the Siegel modular form $\Delta_1$ from \cite{Gritsenko:1996tm}:
\begin{equation}
    \Phi_{3 A, 3 A_3}\left(\Omega\right)=\Delta_1(\Omega)^2.
\end{equation}
It satisfies 
\begin{equation}
    \Phi_{3 A, 3 A_3}\left(V_3 \cdot \Omega\right)=\Phi_{3 A, 3 A_3}(\Omega) .
\end{equation}
The dyon partition function is thus a Siegel modular form of weight 1 for $\Gamma_3^+(1)$ with multiplier $\nu_{g_M,g_N}$ given by
\begin{equation}
\begin{split}
\nu_{g_M,g_N}(g^1(0,-1,1,0)) = \nu_{g_M,g_N}(V_3) = \nu_{g_M,g_N}(g^3(\lambda,\mu,\kappa))&=1, \\ \nu_{g_M,g_N}(g^1(1,1,0,1)) &=\bold{e}^{-1/3},\\ \nu_{g_M,g_N}(g^3(0,0,\kappa/3)) &= \bold{e}^{\kappa/3}.
\end{split}
\end{equation}
Again one can show that the character is trivial on elements of the form \eqref{eq:inv_elements_threshold1} and \eqref{eq:inv_elements_threshold2} using similar analysis as in \eqref{eq:multiplier_def} -- \eqref{eq:trivial_mult_final}.
\par
The dyon partition function admits an additive lift with seed function \cite{Persson:2013xpa}
\begin{equation}
    \psi_{g_M, g_N}(\tau, z)=-\vartheta_1(\tau, z)^2 \eta(\tau)^2.
\end{equation}
 To the best of our knowledge the above result is new and has not been reported earlier in the literature.
 
\section{Summary and discussion}\label{sec:5}
Let us summarize the results of the present paper. We studied $\Z_N$-invariant states in $\Z_M$ CHL models obtained by orbifolding type IIB string theory on $\CM\times T^2$ ($\CM=\text{K}3,T^4$) with a $\Z_M\times\Z_N$ symmetry at some point in the moduli space. These also include theories at those points in the moduli space where we have a $\Z_{MN}$ symmetry with gcd$(M,N)=1$ since $\Z_{MN}\cong\Z_M\times\Z_N$ in this case. Starting with the partition function for $\Z_N$-invariant quarter BPS dyonic states in $\Z_M$ CHL models, we explored its transformation properties. In particular, we found the invariance properties of the associated Siegel modular form $\Phi^{(M,N)}_k$. We found special generators of the subgroup $G < \rm{Sp}(4,\R)$ under which $\Phi^{(M,N)}_k$ transforms as a Siegel modular form. This is instrumental in deriving the modular properties of the degeneracy function. Next we found the Fourier-Jacobi coefficients of the partition function, which are meromorphic Jacobi forms and split into finite parts, which are mock Jacobi forms, and polar parts. From the mock Jacobi part, we extracted the degeneracy of the single-centered black holes. There are two independent Fourier decompositions of the Siegel modular form $\Phi^{(M,N)}_k$ as given in equation (\ref{eq:Fourier_Jacobi_PhiMN}). These two decompositions provide us with two distinct mock Jacobi forms. Depending on the range of charges, the degeneracies of the single-centered black holes can be found from either of the two functions. One can use the $S$-duality transformation to extend the range of charges for which the degeneracy can be computed using mock Jacobi forms. 

There are two general ways of constructing Siegel modular forms: additive lifts and multiplicative lifts. The starting point of both these is a Jacobi form called the additive or multiplicative seed function. Some Siegel modular forms admit both an additive and a multiplicative lift construction. Although the original construction of the partition function of quarter BPS states in CHL model was via additive lift, their construction by precision counting of microstates \cite{sen2008black} is most naturally expressed in terms of an infinite product. For the twisted-twined BPS states, the partition function again is given in terms of an infinite product form, see \eqref{eq:product_PhiMN}. In this paper, we described a method of expressing the partition function as an additive lift complementing and generalizing the result of \cite{govindarajan2011bkm}. To do so, we extensively made use of the multiplicative lift construction of Siegel modular forms $\Phi_{g,h}$ associated to any commuting pair $(g,h)$ of the Mathieu group \cite{Persson:2013xpa}. This construction of Siegel modular forms has the advantage that if there exists an additive lift construction of $\Phi_{g,h}$, then the additive seed is straightforward to identify. We showed that, for $\CM=\text{K}3$, the partition function $\Phi^{(M,N)}_k$ of twisted-twined quarter BPS states is equal to $\Phi_{g_M,g_N}$, where $(g_M,g_N)$ is the generator of the $\Z_M\times\Z_N$ symplectic automorphism of K3 identified with appropriate elements of the Mathieu group via Mukai's theorem, see section \ref{sec:mult_lift_part_fu} for details. This construction of partition function via multiplicative lift gives us a straightforward method to identify the additive seed, if one exists. We described the additive seeds for the cases $(M,N)=(2,2)$, $(2,4)$, $(4,4)$, $(3,3)$. The first three examples were discussed in \cite{govindarajan2011bkm}, which used the constructions of \cite{Gritsenko:1996ax,clery2011siegel} as possible candidate for the partition function. Our result matches with the additive seed of \cite{govindarajan2011bkm} up to a minus sign. 

We end the paper by discussing some open problems. An obvious next step would be to extend the analysis to the doubly twined sector and identify the charge vectors for which the single-centered degeneracy function will have the mock modular structure. We shall report on the same in future. A promising goal would be to gain further insight into the symmetry of the system, so that we have a better understanding of the partition function. In particular, the charge vectors that do not follow the mock modular structure need to be addressed appropriately. More precisely, it would be desirable to construct the generating function
\begin{equation}
F(\Omega) \coloneqq \sum_{(m,n,\ell)\in\Z^3} 
d(m,n,\ell)\, \bold{e}^{m\tau + n\sigma+\ell z},\nonumber
\end{equation}
of single-centered degeneracies $d(m,n,\ell)$ defined by \eqref{degeneracyformula}, and study its modular properties, if any. It would also be useful to generalize the constructions of \cite{Bringmann:2012zr} and \cite{rossello2024immortal} and to find the Rademacher series expansion for the coefficients of mock Jacobi forms for the general $(M,N)$ case. The final goal would be to prove Sen's positivity conjecture for the general $(M,N)$ case.   

\acknowledgments{
We thank Ashoke Sen and Suresh Govindarajan for  discussions on the project. We would also like to thank the anonymous referee for comments which improved certain parts of the paper. VB would like to thank IISER Bhopal for hospitality where part of this work was done. NB would like to acknowledge the hospitality of ICTP during the final stage of the work. The work of NB is supported by SERB POWER fellowship SPG/2022/000370. The work of RKS is supported by the US Department of Energy under grant DE-SC0010008. We thank people of India for their generous support to the advancement of basic sciences.}

\appendix
\section{Invariance groups for $\Phi^{(M,N)}_k$ from threshold integral}\label{app:threshold_integral}

In this appendix, we find certain generators of the group $G$ under which $\Phi^{(M,N)}_k$ is a Siegel modular form. To do so we use the threshold integral expression for the partition function which was given in \cite{sen2010discrete}. For $\rho\in\mathds{H}$, define,
\begin{equation}
    h_b^{(r, s ; r', s')}(\rho)=\sum_{n \in \frac{1}{M N} \mathds{Z}-\frac{b^2}{4}} c_b^{(r, s ; r', s')}(4 n) \, \bold{e}^{n \rho},
\end{equation}
where $c_b^{(r,s;r',s')}(n)$ is as in \eqref{eq:cbrsr's'def},
and,
\begin{equation}
\begin{aligned}
& \frac{1}{2} p_R^2=\frac{1}{4 \det \Img \Omega}\left|-m_1 \tau+m_2+n_1 \sigma+n_2\left(\sigma \tau-z^2\right)+j z\right|^2, \\
& \frac{1}{2} p_L^2=\frac{1}{2} p_R^2+m_1 n_1+m_2 n_2+\frac{1}{4} j^2 .
\end{aligned}
\end{equation}

The \textit{threshold integral} is then defined by 
\begin{equation}
    \widetilde{\mathcal{I}}(\Omega) = \sum_{b=0}^1 \sum_{r, s=0}^{N-1} \sum_{r^{\prime}, s^{\prime}=0}^{M-1} \widetilde{\mathcal{I}}_{r, s ; r^{\prime}, s^{\prime} ; b}(\Omega),
\end{equation}
where,
\begin{equation}\label{eq:threshold_sum}
\begin{aligned}
\widetilde{\mathcal{I}}_{r, s ; r^{\prime}, s^{\prime} ; b} (\Omega) = & \int_{\mathcal{F}} \frac{d^2 \rho}{\Img\rho}\bigg[\sum_{\substack{m_1 \in \mathds{Z}, m_2 \in \mathds{Z} / N, \\
n_1 \in \mathds{Z}+\frac{r^{\prime}}{M}, n_2 \in N \mathds{Z}-r, \\ \\ j \in 2 \mathds{Z}+b}} \bold{e}^{m_1 s^{\prime} / M} \bold{e}^{-m_2 s} \, \mathfrak{q}^{p_L^2 / 2} \overline{\mathfrak{q}}^{\, p_R^2 / 2} h_b^{(r, s ; r^{\prime}, s^{\prime})}(\rho). \\
& \hskip 3in-\delta_{b, 0} \delta_{r, 0} \delta_{r^{\prime}, 0} \, c_0^{(0, s ; 0, s^{\prime})}(0)\bigg].
\end{aligned}
\end{equation} and $\mathfrak{q} \coloneqq \bold{e}^{\rho}$, $\mathcal{F}$ is a fundamental domain of $\slz$ in $\hh$. 

Then one can show that \cite{David:2006yn,sen2010discrete}, 
\begin{equation}
    \widetilde{\mathcal{I}}(\Omega)=-2 \ln \left[(\det \Img \Omega)^{k} \, \Phi^{(M,N)}_k(\Omega)\, \overline{\Phi^{(M,N)}_k(\Omega)} \right] +\text{constant.}
\end{equation}
Suppose $g\in\mathrm{Sp}(4,\Z)$ such that $g$ leaves $\widetilde{\mathcal{I}}$ invariant. Then using the relation
\begin{equation}
    \mathrm{Im}(g\cdot\Omega)=(C\overline{\Omega}+D)^{-1}\mathrm{Im}(\Omega)(C\Omega+D)^{-1},
\end{equation}
it is easy to show that 
\begin{equation}
    \Phi^{(M,N)}_k(g\cdot\Omega)=\mathrm{det}(C\Omega+D)^k\Phi^{(M,N)}_k(\Omega),
\end{equation}
that is, $g \in G$, the subgroup of $\mathrm{Sp}(4,\Z)$ over which $\Phi^{(M,N)}_k$ is a Siegel modular form. Thus, our strategy will be to find elements of $\mathrm{Sp}(4,\Z)$ that leave $\widetilde{\mathcal{I}}(\Omega)$ invariant. To do so, we introduce the matrix
\begin{equation}\label{eq:lattice_matrix}
\mathscr{M}:=\begin{pmatrix}
0 & -m_2 & \frac{j}{2} & n_1 \\
m_2 & 0 & m_1 & -\frac{j}{2} \\
-\frac{j}{2} & -m_1 & 0 & -n_2 \\
-n_1 & \frac{j}{2} & n_2 & 0
\end{pmatrix},
\end{equation}
and 
\begin{equation}
\begin{split}
    D(\mathscr{M},\Omega)&:=-m_1 \tau+m_2+n_1 \sigma+n_2\left(\sigma \tau-z^2\right)+j z,
    \\
    \Delta(\mathscr{M})&:=m_1n_1+m_2n_2+\frac{1}{4}j^2.
\end{split}    
\end{equation}
Then it is a straightforward calculation to show that \cite{Murthy:2009dq}
\begin{equation}
\begin{split}
    D(g\mathscr{M}g^T, g\cdot\Omega)&=[\mathrm{det}(C \Omega+D)]^{-1} D(\mathscr{M}, \Omega),\\
    \Delta(g\mathscr{M}g^T)&=\Delta(\mathscr{M}),\qquad g=\begin{pmatrix}
        A&B\\C&D
    \end{pmatrix}\in\mathrm{Sp}(4,\Z).
\end{split}    
\end{equation}
We use these conditions to test potential candidates for generators of $G$. 

Let $\Lambda$ be the set of all matrices $\mathscr{M}$ of the form \eqref{eq:lattice_matrix} such that the entries of the matrix are constrained as in the indices of the sum \eqref{eq:threshold_sum}:
\begin{equation}
    \Lambda:=\left\{\begin{pmatrix}
0 & -m_2 & \frac{j}{2} & n_1 \\
m_2 & 0 & m_1 & -\frac{j}{2} \\
-\frac{j}{2} & -m_1 & 0 & -n_2 \\
-n_1 & \frac{j}{2} & n_2 & 0
\end{pmatrix}:\begin{array}{c}m_1 \in \mathds{Z}, \, m_2 \in \mathds{Z} / N, \, n_1 \in \mathds{Z}+\frac{r^{\prime}}{M} \\
n_2 \in N \mathds{Z}-r, j \in 2 \mathds{Z}+b\end{array}\right\}
\end{equation}
It is clear then that the threshold integral remains invariant under $g\in\mathrm{Sp}(4,\Z)$ if $g$ preserves $\Lambda$ as well as preserves $m_1\pmod M$ and $m_2\pmod 1$. The latter condition makes sure that the phases $\bold{e}^{m_1s'/M}$ and $\bold{e}^{-m_2s}$ in the threshold integral remain invariant. 

We now check this for the elements listed in \eqref{eq:g_123def}, one by one:
\begin{enumerate}
\item We have
\begin{equation}
    g_1(a,b,cP,d)\mathscr{M}g_1(a,b,cP,d)^t=\begin{pmatrix}
0 & -m_1 b-a m_2 & \frac{(a d-bcP) j}{2} & -n_2 b+a n_1 \\
m_1 b+a m_2 & 0 & c m_2 P+d m_1 & -\frac{j}{2} \\
\frac{-j(ad-bcP)}{2} & -c m_2 P-d m_1 & 0 & c n_1 P-d n_2 \\
n_2 b-a n_1 & \frac{j}{2} & -c n_1 P+d n_2 & 0
\end{pmatrix}.
\end{equation}
If we take $P=MN$, then $g_1$ satisfies all the constraints. 
\item We have 
\begin{equation}
g^1(a,bP,c,d)\mathscr{M}g^1(a,bP,c,d)^T=\begin{pmatrix}
0 & n_1 b P-a m_2 & \frac{j}{2} & d n_1-c m_2 \\
a m_2-n_1 b P & 0 & a m_1+n_2 b P & \frac{-j(ad-bcP)}{2} \\
-\frac{j}{2} & -a m_1-n_2 b P & 0 & -c m_1-d n_2 \\
c m_2-d n_1 & \frac{j(ad-bcP)}{2} & c m_1+d n_2 & 0
\end{pmatrix}.
\end{equation}
If we take $P=M$ and $c\in N\Z$, then all constraints are satisfied. 
\item We have 
{\tiny\begin{equation}
\begin{split}
&g_4(a,b,cP,d)\mathscr{M}g_4(a,b,cP,d)^T  =\\  &\begin{pmatrix}
0 & (-a d+b c P) m_2 & \frac{1}{2}\left(a d j+b c j P+2 b d m_1-2 a c P n_1\right) & -b\left(a j+b m_1\right)+a^2 n_1 \\
(a d-b c P) m_2 & 0 & d^2 m_1+c P\left(d j-c P n_1\right) & -\frac{1}{2} j(a d+b c P)-b d m_1+a c P n_1 \\
-\frac{1}{2} j(a d+b c P)-b d m_1+a c P n_1 & -d^2 m_1+c P\left(-d j+c P n_1\right) & 0 & (-a d+b c P) n_2 \\
b^2 m_1+a\left(b j-a n_1\right) & \frac{1}{2}\left(a d j+b c j P+2 b d m_1-2 a c P n_1\right) & (a d-b c P) n_2 & 0
\end{pmatrix}.
\end{split}
\end{equation}}
If we take $P=M$ then all constraints are satisfied. 
\item We have 
{\tiny\begin{equation}
\begin{split}
    &g_3(\lambda,\mu,\kappa)\mathscr{M}g_3(\lambda,\mu,\kappa)^T  =\\  &\begin{pmatrix}
0 & -m_2+(\kappa-\lambda \mu)n_1+\mu\left(j+\mu n_2\right) & \frac{j}{2}-\lambda n_1+\mu n_2 & n_1\\
m_2+(-\kappa+\lambda \mu) n_1-\mu\left(j+\mu n_2\right) & 0 & m_1+\lambda\left(j-\lambda n_1\right)+(\kappa+\lambda \mu) n_2 & -\frac{j}{2}+\lambda n_1-\mu n_2 \\
-\frac{j}{2}+\lambda n_1-\mu n_2 & -j \lambda-m_1+\lambda^2 n_1-(\kappa+\lambda \mu) n_2 & 0 & -n_2 \\
-n_1 & \frac{j}{2}-\lambda n_1+\mu n_2 & n_2 & 0
\end{pmatrix}.
\end{split}
\end{equation}}
If we take $\mu\in\Z$ and $\lambda,\kappa\in M\Z$, then all constraints are satisfied by $g_3$.
\item We have 
{\tiny\begin{equation}
\begin{split}
  &g^3(\lambda,\mu,\kappa)\mathscr{M}g^3(\lambda,\mu,\kappa)^T  =\\  &  \begin{pmatrix}
0 & (-\kappa+\lambda \mu) m_1-m_2+\mu\left(j+\mu n_2\right) & \frac{j}{2}+\lambda m_1+\mu n_2 & -j \lambda-\lambda^2 m_1+n_1-(\kappa+\lambda \mu) n_2 \\
(\kappa-\lambda \mu) m_1+m_2-\mu\left(j+\mu n_2\right) & 0 & m_1 & -\frac{j}{2}-\lambda m_1-\mu n_2 \\
-\frac{j}{2}-\lambda m_1-\mu n_2 & -m_1 & 0 & -n_2 \\
j \lambda+\lambda^2 m_1-n_1+(\kappa+\lambda \mu) n_2 & \frac{j}{2}+\lambda m_1+\mu n_2 & n_2 & 0
\end{pmatrix}.
\end{split}
\end{equation}}
If we take $\lambda,\mu,\kappa\in\Z$, then all constraints are satisfied by $g^3$.
\end{enumerate}

Then by the embeddings (\ref{eq:sp4_zembeddings}), we have shown that,
\begin{itemize}
    \item $g_1(\Gamma(MN)) \subset G$,
    \item $g^1(\Gamma_1(N) \cap \Gamma^1(M)) \subset G$,
    \item $g_4(\Gamma_1(M)) \subset G$,
    \item $g_3(\lambda,\mu,\kappa) \in G$ for all $\lambda, \kappa \in M\Z$ and $\mu \in \Z$, 
    \item $g^3(\lambda,\mu,\kappa) \in G$ for all $\lambda, \mu \kappa \in \Z$.
\end{itemize}

This proves theorem \ref{thm:invariancegroup_PhiMN}.

\section{Deformation of contours}
The following proposition lets us determine the charge sector for which we can deform the attractor contour $\mathcal{C}$ to contour $\mathcal{C}'$

\begin{prop}
\label{prop:deformationofcontour}
    If $(n,\ell,m) \in A$, then the twisted-twined partition function ${\Phi^{M,N}_k}^{-1}$ does not have poles in the region bounded by the attractor contour $\mathcal{C}$ and the deformed contour $\mathcal{C}'$ given by $\tau \to i \infty$.
\end{prop}

\begin{proof}
For small $\varepsilon$, as is the case for the attractor contour, the dominant contribution to the LHS of the pole equation (\ref{eq:zeroes}) is given by the term $-n_2 (\Img(\sigma) \Img(\tau) - \Img(z)^2)$. The term inside the parenthesis is always positive on $\hh_2$, ensuring that the coefficient of $n_2$ does not vanish. The above is true for the entire region between $\mathcal{C}$ and $\mathcal{C}'$. Therefore, we must set $n_2=0$ to satisfy the pole equation. Thus, we have the simplified condition,
\begin{equation}
\label{eq:polesatinfinity}
    n_1 \sigma - m_1 \tau + jz + m_2 = 0,
\end{equation} where,
\begin{equation}
    m_1 \in M \Z, \qquad n_1, m_2 \in \Z, \qquad j \in 2\Z + 1, \qquad m_1 n_1 + \frac{j^2}{4} = \frac{1}{4}.
\end{equation} 
As $n_2$ carried the information about the twisting in the pole structure, setting it to zero in the region of interest implies that the pole structure for our discussion is the same as that of the un-twisted case, and the proof of the analogous theorem in \cite{bhand2023mock} follows through. However, for completeness, we discuss the rest of the proof here too.

The imaginary parts of the integration variables is fixed on the contour. So from equation (\ref{eq:polesatinfinity}) we find,
\begin{equation}
    n_1 \Img(\sigma) - m_1 \Img(\tau) + j\Img(z) = 0,
\end{equation} which implies,
\begin{equation}
    \Img(\tau) = \frac{2 m n_1 - j \ell}{m_1 \varepsilon}.
\end{equation} The above describes the locations of the poles, so to ensure that there are no poles between $\mathcal{C}$ and $\mathcal{C}'$, we demand,
\begin{equation}
    \frac{2 m n_1 - j \ell}{m_1 \varepsilon} \leq \frac{2n}{M\varepsilon},
\end{equation} where the RHS is the value of the imaginary part of $\tau$ for the attractor contour.\footnote{For the equality, ie, the attractor point lying on the wall of marginal stability, the jump in the index on crossing the wall vanishes, ie, the residue vanishes. This is shown in \cite[Appendix B.2]{bhand2023mock}} This simplifies to the expression,
\begin{align}
    \bigg(j + \ell\frac{m_1}{m}\bigg)^2 + \frac{m_1^2}{m^2}\bigg(\frac{4nm}{M} - \ell^2 \bigg) \geq 1,
\end{align} which is true for all $m_1 \in M\Z$, $j \in 2\Z + 1$, $m_1 |(j^2 -1)/4$. This can be written as,
\begin{align}
    \min_{r \in M\Z, \, s\in \Z, \, s(s+1) = 0 \, (\mathrm{mod} \, {r})}\bigg[\bigg(\frac{r\ell}{m} + 2s + 1\bigg)^2 + \frac{r^2}{m^2}\bigg(\frac{4nm}{M}-\ell^2 \bigg)\bigg] \geq 1,
\end{align} where we have reparametrized $m_1 = r$, and $j = 2s+1$ with $s \in \Z$. Furthermore, single centered black holes exist only for charges $(\vec{Q}, \vec{P})$ such that $Q^2 P^2 - (Q \cdot P)^2 > 0$, which translates to
\begin{equation}
    \frac{4nm}{M} - \ell^2 > 0.
\end{equation} This proves the theorem.
\end{proof}

\begin{prop}
\label{prop:deformationofcontour2}
    If $(n,\ell,m) \in B$, then the twisted-twined partition function ${\Phi^{M,N}_k}^{-1}$ does not have poles in the region bounded by the attractor contour $\mathcal{C}$ and the deformed contour $\mathcal{C}'$ given by $\sigma \to i \infty$.
\end{prop}
\begin{proof}
    The proof is analogous to the above proposition. $n_2$ has to still be set to zero for the same reason as before. The condition for there to be no poles between $\mathcal{C}$ and $\mathcal{C}''$ turns out to be,
\begin{equation}
    \frac{m_1 + j \frac{M\ell}{2n}}{n_1} \leq \frac{Mm}{n},
\end{equation} which simplifies to,
\begin{align}
    \bigg(j + \ell\frac{n_1 M}{n}\bigg)^2 + \bigg(\frac{n_1 M}{n} \bigg)^2\bigg(\frac{4nm}{M} - \ell^2 \bigg) \geq 1.
\end{align} Demanding that this is true for all $j$ and $n_1$ gives us the condition,
\begin{equation}
    \min_{r \in M\Z, \, s \in \Z, \, s(s+1) = 0 \, (\mathrm{mod} \, {r/M})}\bigg[\bigg(\frac{r\ell}{n} - 2s - 1\bigg)^2 + \frac{r^2}{n^2}\bigg(\frac{4nm}{M}-\ell^2 \bigg)\bigg] \geq 1.
\end{equation}
\end{proof}

\section{The Mathieu group}
In this appendix, we collect some information about the Mathieu group relevant to the analysis in section \ref{sec:4}. 

The Mathieu group $\IM_{24}$ is one of the sporadic finite simple groups of order 244823040 \cite{ATLAS}. It is the automorphism group of the Golay code \cite{Conway:1988oqe}, that is, a subgroup of the permutation group $S_{24}$ which preserves the Golay code. It has a subgroup $\IM_{23}\subset\IM_{24}$ of order 10200960 which fixes the 24th letter and permutes the first 23 letters \cite{Conway:1988oqe}. 

Recall that a permutation on 24 letters has cycle structure 
\begin{equation}
    1^{n_1}2^{n_2}\cdots 24^{n_{24}}, \qquad \sum_{i=1}^{24}i \cdot n_i=24,
\end{equation}
if the permutation is a product of $n_i$ $i$-cycles. The cycle structure is an invariant of a conjugacy class. $\IM_{24}$ has 27 conjugacy classes which are labeled by the cycle structure of a representative. We list the conjugacy classes in Table \ref{tab:M24_conj_class} below. 

Now, given a commuting pair $g,h\in\IM_{24}$, clearly $h\in C_{\IM_{24}}(g)$. Thus, we can talk about the conjugacy class of $h$ in $C_{\IM_{24}}$. We list the $C_{\IM_{24}}$--classes of $h$ when $g\in 2A$, $3A$, $4B$ in table \ref{tab:conj_class2A}, \ref{tab:conj_class4B} and \ref{tab:conj_class3A} respectively, which are used in section \ref{sec:4}. The rest can be found in \cite[Appendix D]{Gaberdiel:2012gf}. 

We also introduce the notation:
\begin{itemize}
    \item $G=N\cdot Q$: group $G$ containing a normal subgroup $N$ and such that $G / N \cong Q$.
\item $N \rtimes Q $: (left) semidirect product of $N$ and $Q$.
\item $A_n$: group of even permutations (alternating group) of $n$ letters.
\item $\mathrm{PSL}_n(q)$: projective special linear group over the $n$-dimensional vector space over the finite field $\mathds{F}_q$.
\end{itemize}

\begin{table}[H]
\centering
\begin{tabular}{|c|c|c|l|}
\hline
Class & Order & Number of Elements & Cycle Structure \\
\hline
1A  & 1   & 1          & Identity \\
2A  & 2   & 11,385     & $1^2 \cdot 2^8$ \\
2B  & 2   & 31,878     & $2^{12}$ \\
3A  & 3   & 226,688    & $1^6 \cdot 3^6$ \\
3B  & 3   & 485,760    & $3^8$ \\
4A  & 4   & 637,560    & $2^4 \cdot 4^4$ \\
4B  & 4   & 1,912,680  & $1^4 \cdot 2^2 \cdot 4^4$ \\
4C  & 4   & 2,550,240  & $4^6$ \\
5A  & 5   & 4,080,384  & $1^4 \cdot 5^4$ \\
6A  & 6   & 10,200,960 & $1^2 \cdot 2^2 \cdot 3^2 \cdot 6^2$ \\
6B  & 6   & 10,200,960 & $6^4$ \\
7A  & 7   & 5,829,120  & $1^3 \cdot 7^3$ \\
7B  & 7   & 5,829,120  & $1^3 \cdot 7^3$ \\
8A  & 8   & 15,301,440 & $1^2 \cdot 2 \cdot 4 \cdot 8^2$ \\
10A & 10  & 12,241,152 & $2^2 \cdot 10^2$ \\
11A & 11  & 22,256,640 & $1^2 \cdot 11^2$ \\
12A & 12  & 20,401,920 & $2 \cdot 4 \cdot 6 \cdot 12$ \\
12B & 12  & 20,401,920 & $12^2$ \\
14A & 14  & 17,487,360 & $1 \cdot 2 \cdot 7 \cdot 14$ \\
14B & 14  & 17,487,360 & $1 \cdot 2 \cdot 7 \cdot 14$ \\
15A & 15  & 16,321,536 & $1 \cdot 3 \cdot 5 \cdot 15$ \\
15B & 15  & 16,321,536 & $1 \cdot 3 \cdot 5 \cdot 15$ \\
21A & 21  & 11,658,240 & $3 \cdot 21$ \\
21B & 21  & 11,658,240 & $3 \cdot 21$ \\
23A & 23  & 10,644,480 & $1 \cdot 23$ \\
23B & 23  & 10,644,480 & $1 \cdot 23$ \\
24A & 24  & 20,401,920 & $1 \cdot 2 \cdot 3 \cdot 4 \cdot 6 \cdot 8 \cdot 12 \cdot 24$ \\
\hline
\end{tabular}
\caption{Conjugacy classes of the Mathieu group \( \IM_{24} \). The notations 2A,2B,...  for the conjugacy class comes from the largest $n$ such that the cycle structure has an $n$-cycle.}
\label{tab:M24_conj_class}
\end{table}

\begin{table}[H]
\centering
\begin{tabular}{|c|c|l|}
\hline
Class & Order of Centralizer \( C_{\IM_{24}}(g) \) & Group Structure of Centralizer \\
\hline
2A & 21504  & \( \mathds{Z}_2^4 \cdot\left(\mathds{Z}_2^3 \rtimes \mathrm{P S L}_3(2)\right) \) \\
3A & 1080   & \( \mathds{Z}_3 \cdot A_{6} \) \\
4B & 128   & \( \left(\left(\mathds{Z}_4 \times \mathds{Z}_4\right) \rtimes \mathds{Z}_4\right) \rtimes \mathds{Z}_2 \)  \\
\hline
\end{tabular}
\caption{Centralizers \( C_{\IM_{24}}(g) \) for representatives \( g \in 2A, 3A, 4B \) in \( \IM_{24} \). The result is taken from \cite[Appendix D]{Gaberdiel:2012gf}.}
\label{tab:cent_2A3A4B}
\end{table}


\begin{table}[H]
\centering
\begin{minipage}[t]{0.48\textwidth}
\centering
\begin{tabular}{|c|c|}
\hline
Order & Conjugacy Class \\
\hline
21504 & $1A$ \\
21504 & $2A_1$ \\
1536  & $2A_2$ \\
1536  & $2A_3$ \\
384   & $4A_1$ \\
512   & $2B_1$ \\
256   & $2A_4$ \\
256   & $2B_3$ \\
64    & $4B_2$ \\
128   & $2A_5$ \\
128   & $2B_2$ \\
64    & $4A_2$ \\
128   & $4B_1$ \\
64    & $4A_3$ \\
32    & $4C_1$ \\
64    & $4B_3$ \\
32    & $4B_4$ \\
32    & $4A_4$ \\
16    & $4B_5$ \\
16    & $4C_2$ \\
16    & $8A_1$ \\
24    & $3A_1$ \\
24    & $6A_1$ \\
12    & $6A_2$ \\
12    & $6A_3$ \\
12    & $12A_1$ \\
14    & $7A_1$ \\
14    & $14A_1$ \\
14    & $7B_1$ \\
14    & $14B_1$ \\
\hline
\end{tabular}
\caption{Conjugacy classes of $C_{\IM_{24}}(g)$ for $g\in 2A$. The result is taken from \cite[Appendix D]{Gaberdiel:2012gf}.
\label{tab:conj_class2A}}
\end{minipage}%
\hfill
\begin{minipage}[t]{0.48\textwidth}
\centering
\begin{tabular}{|c|c|}
\hline
Order & Conjugacy Class \\
\hline
128 & $1A_1$ \\
16  & $8A_1$ \\
128 & $4B_2$ \\
128 & $4B_6$ \\
64  & $2B_3$ \\
64  & $4B_1$ \\
128 & $2A_1$ \\
16  & $2B_1$ \\
16  & $4B_3$ \\
16  & $2A_3$ \\
16  & $4B_5$ \\
32  & $2B_2$ \\
16  & $8A_2$ \\
16  & $8A_3$ \\
32  & $4A_2$ \\
32  & $4A_1$ \\
32  & $2A_4$ \\
16  & $4B_7$ \\
16  & $4B_4$ \\
16  & $4A_3$ \\
16  & $4A_4$ \\
64  & $2A_2$ \\
16  & $8A_4$ \\
64  & $4B_9$ \\
64  & $2A_5$ \\
64  & $4B_8$ \\
\hline
\end{tabular}
\caption{Conjugacy classes of $C_{\IM_{24}}(g)$ for $g\in 4B$. The result is taken from \cite[Appendix D]{Gaberdiel:2012gf}.
\label{tab:conj_class4B}}
\end{minipage}
\end{table}
\begin{table}[H]
\centering
\resizebox{\textwidth}{!}{
\begin{tabular}{|r|r|r|r|r|r|r|r|r|r|r|r|r|r|r|r|r|r|}
\hline
Order & 1080 & 1080 & 1080 & 24 & 24 & 24 & 9 & 9 & 12 & 12 & 12 & 15 & 15 & 15 & 15 & 15 & 15 \\
\hline
Conjugacy Class & $1A$ & $3A_1$ & $3A_2$ & $6A_1$ & $6A_2$ & $2A_1$ & $3A_3$ & $3B_1$ & $12A_1$ & $12A_2$ & $4A_1$ & $15A_1$ & $5A_1$ & $15B_1$ & $5A_2$ & $15A_2$ & $15B_2$ \\
\hline
\end{tabular}
}
\caption{Conjugacy classes of $C_{\IM_{24}}(g)$ for $g\in 3A$. The result is taken from \cite[Appendix D]{Gaberdiel:2012gf}.
}
\label{tab:conj_class3A}
\end{table}

\bibliographystyle{unsrt}
\bibliography{citations.bib}	
\end{document}